\documentclass[12pt,a4paper,twoside]{amsart} 
\usepackage{latexsym}
\usepackage{amsmath}
\usepackage{amssymb}
\usepackage{amsfonts}
\usepackage{ifthen}
\usepackage{mathrsfs}               
\usepackage{bbm}                    
\usepackage{bbold}                  
\usepackage{textcomp}
\usepackage{amsthm}
\usepackage{amstext}
\usepackage{enumerate}
\newcommand{\CC}{\mathbb{C}}

\newcommand{\NN}{\mathbb{N}}

\newcommand{\RR}{\mathbb{R}}

\newcommand{\supp}{\mathrm{supp}}
\newcommand{\const}{\mathrm{const}}
\newcommand{\dist}{\mathrm{dist}}
\newcommand{\Ran}{\mathrm{Ran}}

\newcommand{\Ker}{\mathrm{Ker}} 

\newcommand{\sgn}{\mathrm{sgn}}
\newcommand{\loc}{\mathrm{loc}}

\newcommand{\rot}{\mathrm{curl}}
\newcommand{\Lip}{\mathrm{Lip}}
\newcommand{\id}{\mathbbm{1}}
\newcommand{\klg}{\leqslant} 
\newcommand{\grg}{\geqslant}          
\newcommand{\ve}{\varepsilon}
\newcommand{\vp}{\varphi}
\newcommand{\vk}{\varkappa}
\newcommand{\vr}{\varrho}
\newcommand{\vt}{\vartheta}

\newcommand{\wt}[1]{\widetilde{#1}}
    
\newcommand{\SPn}[2]{\langle \,#1\,|\,#2\, \rangle} 
\newcommand{\SPb}[2]{\big\langle \,#1\,\big|\,#2\, \big\rangle} 
\newcommand{\SPB}[2]{\Big\langle \,#1\,\Big|\,#2\, \Big\rangle}
\newcommand{\ol}[1]{\overline{#1}} 
\newcommand{\bigO}{\mathcal{O}}    
\newcommand{\DA}{D_{A}}             
\newcommand{\DAV}{D_{A,V}}
\newcommand{\D}[1]{D_{#1}}
\newcommand{\B}[1]{B_{#1}}          
\newcommand{\BAVC}{B_{A,V}}
\newcommand{\PO}{\Lambda^+_{0}}        
\newcommand{\PA}{\Lambda^+_{A}}
\newcommand{\PAF}{\Lambda^F_{A}}
\newcommand{\PAFm}{\Lambda^{-F}_{A}}
\newcommand{\PAV}{\Lambda^+_{A,V}}
\newcommand{\PAm}{\Lambda^-_{A}}
\newcommand{\POm}{\Lambda^-_{0}}
\newcommand{\RA}[1]{R_A(#1)}       

\newcommand{\RAV}[1]{R_{A,V}(#1)}
\newcommand{\R}[2]{R_{#1}(#2)}
\newcommand{\dom}{\mathcal{D}}  
\newcommand{\ball}[2]{\mathcal{B}_{#1}(#2)} 
\newcommand{\LO}{\mathscr{L}}      
\newcommand{\COI}{\mathscr{D}}     
\newcommand{\HR}{\mathscr{H}}      
\newcommand{\spec}{\mathrm{\sigma}}
\newcommand{\specac}{\mathrm{\sigma}_{\mathrm{ac}}}

\newcommand{\specess}{\mathrm{\sigma}_{\mathrm{ess}}}
\newcommand{\specd}{\mathrm{\sigma}_{\mathrm{d}}}

\newcommand{\cQ}{\mathcal{ Q}}

\newcommand{\cU}{\mathcal{ U}}

\newcommand{\sA}{\mathscr{A}}

\newtheorem{theorem}{Theorem}[section]
\newtheorem{lemma}[theorem]{Lemma}

\newtheorem{corollary}[theorem]{Corollary}

\newtheorem{example}[theorem]{Example}
\newtheorem{remark}[theorem]{Remark}
\newtheorem{hypothesis}{Hypothesis}
\title[Eigenfunctions
of no-pair operators]{ On the eigenfunctions
of no-pair operators
in classical magnetic fields}
\author{Oliver Matte and Edgardo Stockmeyer}
\address{
Oliver Matte\\
Mathematisches Institut\\
Ludwig-Maximilians-Universit\"at\\Theresienstra{\ss}e 39\\
D-80333 M\"unchen, Germany.}
\email{matte@math.lmu.de}
\address{
Edgardo Stockmeyer\\
Institut f\"ur Mathematik\\
Johannes Gutenberg-Uni- versit\"at\\
Staudingerweg 9\\
D-55099 Mainz, Germany. {\it On leave from:} Mathemati- sches Institut\\
Ludwig-Maximilians-Universit\"at\\Theresienstra{\ss}e 39\\
D-80333 M\"unchen, Germany.}
\email{stock@math.lmu.de}
\begin{document}

\begin{abstract}
\noindent
We consider a relativistic no-pair model of a
hydrogenic atom in a classical, exterior magnetic field. 
First, we prove that the corresponding Hamiltonian is
semi-bounded below, for
all coupling constants less
than or equal to the critical one known for
the Brown-Ravenhall model, i.e., for
vanishing magnetic fields.
We give conditions ensuring that its essential spectrum equals
$[1,\infty)$ and that there exist infinitely many
eigenvalues below $1$. (The rest energy of the electron is $1$
in our units.) Assuming that the magnetic vector potential is smooth
and that all its partial derivatives increase subexponentially,
we finally show that an eigenfunction corresponding
to an eigenvalue $\lambda<1$ is smooth away from the nucleus
and that its partial derivatives of any order
decay pointwise exponentially with any rate
$a<\sqrt{1-\lambda^2}$, for $\lambda\in[0,1)$, and $a<1$, for $\lambda<0$.

\smallskip

\noindent{\bf Keywords:} 
Dirac operator, Brown and Ravenhall, 
no-pair operator, 
exponential decay, regularity.
\end{abstract}

\maketitle


\section{Introduction}

\noindent
The aim of this article is to study the regularity and the
pointwise exponential decay
of eigenstates of relativistic 
hydrogenic atoms in exterior magnetic fields
which are described in the {\em free picture}. 
The latter model 
is obtained by restricting the usual Coulomb-Dirac operator
with magnetic vector potential, $A$,
to the positive spectral subspace
of the magnetic Dirac operator {\em without} electrostatic potential.
We shall call the
resulting operator the no-pair operator,
since it belongs to a more general class of models which can be derived
by a formal procedure in quantum electrodynamics that neglects
pair creation and annihilation processes \cite{Su1,Su2}.
If we set $A=0$, then the no-pair
operator considered here is also known as the (one-particle) Brown-Ravenhall 
or Bethe-Salpeter operator 
\cite{BeSa,BrRa}. 
(Numerous mathematical contributions to the Brown-Ravenhall model
are listed in the references to \cite{MaSt08a}.)
Although these models have their
main applications in the numerical
study of relativistic atoms with a large number of
electrons \cite{Su1,Su2}, they pose some new mathematical
problems already in the investigation of
hydrogenic atoms.
This is due to the fact that both the kinetic
and the potential part of the no-pair operator
are non-local.
There already exist results on the
$L^2$-exponential localization of bound states of (multi-particle)
Brown-Ravenhall operators.
All of them give, however, suboptimal bounds on the decay rate.
The first one has been derived in \cite{BaMa}
for a hydrogenic atom and
for coupling constants less than $1/2$.
It has been generalized in \cite{Mo} to many-electron atoms
and to all coupling constants below and including
the critical one of the Brown-Ravenhall model
determined in \cite{EPS}.
In \cite{MaSt08a} the present authors
study a no-pair model of a 
many-electron atom which is defined by means
of projections {\em including} the electrostatic potential
as well as perhaps a mean-field and a non-local exchange
potential. The main results of \cite{MaSt08a}
are an HVZ theorem, conditions for the existence of infinitely many
discrete eigenvalues, and
$L^2$-estimates on the exponential localization
of the corresponding eigenvectors.

Besides the passage to
pointwise exponential bounds on the
partial derivatives of eigenstates of the no-pair operator
(for a class of magnetic vector potentials whose partial
derivatives of any order are allowed to increase subexponentially),
the present article 
includes some further improvements, even in the case $A=0$.
First, we verify that the rate of exponential decay
of an eigenvector of the no-pair operator
corresponding to an eigenvalue $\lambda<1$
is not less than any 
\begin{equation}\label{def-triangle}
a\,<\,\triangle(\lambda)\,:=\,
\left
\{\begin{array}{ll}
\sqrt{1-\lambda^2},&\lambda\in[0,1),
\\
1,&\lambda<0\,.
\end{array}
\right.
\end{equation}
This is the same behaviour as it is known for the
Chandrasekhar operator \cite{Ca,CaMaSi,HePa}.
We remark that the Brown-Ravenhall operator is strictly
positive \cite{Tix1}. The lowest eigenvalue of the no-pair operator,
however, is expected to
tend to $-\infty$ as the strength of a constant exterior magnetic field
is increased; see \cite{J-A} for some numerical evidence.
Secondly, in order to find a distinguished self-adjoint realization of
the no-pair operator
we show that the corresponding quadratic
form is bounded from below, for all coupling constants
less than or equal to the critical one of the Brown-Ravenhall model. 
This has been known
before only in the case $A=0$ \cite{EPS} and all
we actually do is to reduce the problem to that special case.
(For smaller values of the coupling constant,
there exist, however, results on the stability of matter of the
second kind in the free picture, where a gauge fixed vector potential
is considered as a variable in the minimization.
In this situation the field energy is added to the multi-particle
Hamiltonian; see \cite{LSS} and \cite{LL} for quantized fields.
It is actually important to include the vector potential in the
projection determining the model for otherwise instability
occurs if at least two electrons are considered \cite{GrTix}.)
Finally, we state conditions ensuring that the essential
spectrum of the no-pair operator equals $[1,\infty)$
and that it has
infinitely many discrete eigenvalues below $1$.

As a byproduct of our analysis -- roughly speaking,
by ignoring the projections --
we find pointwise exponential decay
estimates with a rate $a<\sqrt{1-\lambda^2}$ 
for the eigenfunctions of magnetic Coulomb-Dirac operators
corresponding to an eigenvalue
$\lambda\in(-1,1)$. Although such bounds are essentially
well-known \cite{BeGe,HePa,Wa} it seems illustrative
to include them as a remark here.
For a general scheme to study the exponential decay
of solutions of an elliptic system of partial differential
equations we refer to \cite{RaRo1,RaRo2}.


\section{Definition of the model and main results}

\subsection{The no-pair operator}

If energies are measured in units of the rest energy of the electron
and lengths in units of one Compton
wave length divided by $2\pi$, then the free Dirac operator
is given as
$$
\D{0}\,:=\,-i\,\alpha\cdot\nabla+\beta
\,:=\,-i\sum_{j=1}^3\alpha_j\,\partial_{x_j}\:+\,\beta
\,.
$$
Here $\alpha=(\alpha_1,\alpha_2,\alpha_3)$
and $\beta=:\alpha_0$ are the usual $4\times4$
hermitian Dirac matrices. They are given
as $\alpha_i=\sigma_1\otimes\sigma_i$, $i=1,2,3$,
and $\beta=\sigma_3\otimes\id_2$,
where $\sigma_1,\sigma_2,\sigma_3$
denote the standard Pauli matrices, 
and satisfy the Clifford
algebra relations
\begin{equation}\label{Clifford}
\{\alpha_i\,,\,\alpha_j\}\,=\,2\,\delta_{ij}\,\id\,,
\qquad 0\klg i,j\klg3\,.
\end{equation}
$\D{0}$ is a self-adjoint
operator in the Hilbert space
$$
\HR\,:=\,L^2(\RR^3,\CC^4)
$$
with domain $H^1(\RR^3,\CC^4)$ and its purely absolutely
continuous spectrum equals
$
\spec(\D{0})=
\specac(\D{0})=(-\infty,-1]\cup[1,\infty)\,.
$
Moreover, it is well-known \cite{Ch}
that the free Dirac operator
with magnetic vector potential $A\in L^\infty_\loc(\RR^3,\RR^3)$,
\begin{equation}\label{def-DA}
\DA\,:=\,\D{0}\,+\,\alpha\cdot A
\end{equation}
is essentially self-adjoint on the domain
\begin{equation}\label{def-COI}
\COI\,:=\,C_0^\infty(\RR^3,\CC^4)\,.
\end{equation}
We denote its closure again by the symbol $\DA$.
Its spectrum is again contained
in the union of two half-lines \cite{Th},
$$
\spec(\DA)\,\subset\,(-\infty,-1]\cup[1,\infty)\,.
$$
In order to define the no-pair operator we introduce
the spectral projections
\begin{equation}\label{def-PA}
\PA\,:=\,E_{[0,\infty)}(\DA)\,=\,\frac{1}{2}\,\id\,+\,
\frac{1}{2}\,\sgn(\DA)\,,\qquad \PAm\,:=\,\id-\PA\,,
\end{equation}
and a (matrix-valued) potential, $V$, satisfying
the following hypothesis.

\begin{hypothesis}\label{hyp-Vogel}
$V\in L_\loc^\infty(\RR^3\setminus\{0\},\LO(\CC^4))$,
\begin{equation}\label{VCtozero}
V(x)=V(x)^*\,,\quad x\not=0\,,\qquad
V(x)\,\longrightarrow\,0\,,\quad |x|\to\infty\,.
\end{equation}
and
there exist
$\gamma\in(0,1)$
and $\rho>0$ such that
\begin{equation}\label{eq-Vogel1}
\|V(x)\|_{\LO(\CC^4)}\,\klg\,\frac{\gamma}{|x|}\,,\qquad
0<|x|<\rho\,.
\end{equation}
\end{hypothesis}

\bigskip

\noindent
The no-pair operator is an operator acting in the projected
Hilbert space
\begin{equation}
\HR^+_A\,:=\,\PA\,\HR\,,
\end{equation}
which on the dense subspace $\PA\,\COI$
is given as
\begin{equation}\label{def-BRO}
\B{A,V}\,\vp^+\,:=\,\DA\,\vp^+\,+\,\PA\,V\,\vp^+\,,
\qquad \vp^+\in\PA\COI\,.
\end{equation}
It is not completely obvious that $V\,\PA\,\psi$
is again square-integrable, for every $\psi\in \COI$.
This follows, however, from Lemma~\ref{VPAphi-in-L2}
below.
In order to define a distinguished
self-adjoint realization of $\B{A,V}$ we shall
assume that $V$ satisfies Hypothesis~\ref{hyp-Vogel} with 
$\gamma\klg\gamma_{\mathrm{c}}$, where
\begin{equation}\label{def-gammac}
\gamma_{\mathrm{c}}\,:=\,
\frac{2}{(\pi/2)+(2/\pi)}
\end{equation}
is the critical coupling constant of the Brown-Ravenhall model
determined in \cite{EPS}.
In the case of the atomic Coulomb potential,
$V(x)=-\frac{\gamma}{|x|}\,\id$,
the coupling constant is given by $\gamma=e^2\,Z$,
where $Z\in\NN$ and the square of the electric charge, $e^2$, is equal to
the Sommerfeld fine structure constant in our units,
$e^{-2}\approx137.037$. Since $e^2\,\gamma_{\mathrm{c}}\,\approx\,124.2$
the restriction on the strength of the singularities of $V$
imposed in \eqref{eq-Vogel1} with 
$\gamma<\gamma_{\mathrm{c}}$ or $\gamma\klg\gamma_{\mathrm{c}}$ 
allows for all
nuclear charges up to $Z\klg124$.
It is shown in \cite{EPS} that  
the quadratic form of $\B{0,-\gamma/|\cdot|}$
is bounded below on $\PO\COI$, for all $\gamma\in[0,\gamma_{\mathrm{c}}]$,
and unbounded below if $\gamma>\gamma_{\mathrm{c}}$. 
(Due to a result
of \cite{Tix1} one actually has the strictly positive
lower bound $1-\gamma$, for all
$\gamma\in[0,\gamma_{\mathrm{c}}]$.)
Combining this with
some new technical results on the spectral projections
$\PA$ and $\PO$ derived in Section~\ref{sec-spectral-proj},
we prove the following theorem in Section~\ref{sec-semi-bd}.

\begin{theorem}\label{thm-BRA-bb}
(i) Assume that $V$ fulfills Hypotheses~\ref{hyp-Vogel}
with $\gamma\in(0,\gamma_{\mathrm{c}})$ and
that $A\in L^\infty_\loc(\RR^3,\RR^3)$.
Then
\begin{equation}\label{BR-qf>-infty}
\inf\big\{\,\SPn{\vp^+}{\BAVC\,\vp^+}\::\;
\vp^+\in\PA\COI\,,\;\|\vp^+\|=1\,\big\}\,>\,-\infty\,.
\end{equation}
In particular, by the KLMN-theorem, $\BAVC$ has a
distinguished self-adjoint extension with
form domain 
$\cQ(\DA\!\!\upharpoonright_{\HR^+_A})=\PA\,\dom(|\DA|^{1/2})$.

\smallskip

\noindent
(ii) Assume that $V$ fulfills Hypothesis~\ref{hyp-Vogel}
with $\gamma\in(0,\gamma_{\mathrm{c}}]$ and
that $A\in L^\infty_\loc(\RR^3,\RR^3)$ is Lipschitz continuous in 
some neighbourhood of $0$. Then \eqref{BR-qf>-infty}
holds true also. In particular, $\BAVC$ has a
self-adjoint Friedrichs extension.
\end{theorem}

\bigskip

\noindent
The self-adjoint extension of $\BAVC$ given by
Theorem~\ref{thm-BRA-bb} is again denoted by the same symbol.
We then have the following result.

\begin{theorem}\label{main-thm-L2}
Assume that $V$ fulfills Hypothesis~\ref{hyp-Vogel}
with $\gamma\in(0,\gamma_{\mathrm{c}}]$ and that 
$A\in L_\loc^\infty(\RR^3,\RR^3)$
and let $\triangle:(-\infty,1)\to(0,1]$ be given by~\eqref{def-triangle}.
If $\gamma=\gamma_{\mathrm{c}}$, assume additionally that $A$ is locally
Lipschitz continuous.
Then 
$$
\specess(\BAVC)\,\subset\,[1,\infty)\,,
$$
and for every eigenvector, $\phi_\lambda$, of $\BAVC$
corresponding to an eigenvalue $\lambda<1$ and every
$a<\triangle(\lambda)$,
$$
\big\|\,e^{a|\cdot|}\,\phi_\lambda\,\big\|_\HR\,<\,\infty\,.
$$  
\end{theorem}

\begin{proof}
Theorem~\ref{main-thm-L2} is a consequence of
Theorems~\ref{thm-L2-exp}, \ref{thm-L2-exp-klg0}, and~\ref{thm-lower-bd}.
\end{proof}

\bigskip

\noindent
In order to derive pointwise decay estimates for all
partial derivatives of eigenfunctions
we introduce further assumptions
on $A$ and $V$.

\begin{hypothesis}\label{hyp-AV-smooth}
$A\in C^\infty(\RR^3,\RR^3)$
and, for all $\ve>0$ and $\beta\in\NN_0^3$,
there is some $K(\ve,\beta)\in(0,\infty)$
such that
\begin{eqnarray}
|\partial_x^\beta
A(x)|&\klg&K(\ve,\beta)\,e^{\ve|x|}\,,
\qquad x\in\RR^3\,.\label{exp-est-A-infty}
\end{eqnarray}
$V\in C^\infty(\RR^3\setminus\{0\},\LO(\CC^4))$
fulfills \eqref{VCtozero}\&\eqref{eq-Vogel1}
and, for all $r>0$ and $\beta\in\NN_0^3$, there is 
some $C(r,\beta)\in(0,\infty)$ such that
\begin{equation}\label{hyp-V-Cinfty}
\sup_{|x|\grg r}\big\|\,\partial_x^\beta V(x)\,\big\|_{\LO(\CC^4)}
\,\klg\,C(r,\beta)\,.
\end{equation}
\end{hypothesis}

\bigskip

\noindent
We remark that our $L^2$-exponential bounds on eigenfunctions of $\B{A,V}$
are completely
independent from the behaviour of $A\in L_\loc^\infty(\RR^3,\RR^3)$ 
away from the nucleus; see Theorems~\ref{thm-L2-exp}\&\ref{thm-L2-exp-klg0}.
It seems, however, natural to introduce the condition \eqref{exp-est-A-infty}
to infer the pointwise bounds of Theorem~\ref{main-thm} below
by means of an induction argument starting from Theorem~\ref{main-thm-L2}.
In fact,
since we always consider decay rates which are strictly less than
$\triangle(\lambda)$ we can borrow a bit of the exponential decay
of the eigenfunction $\phi_\lambda$ to control terms containing
a vector potential satisfying \eqref{exp-est-A-infty}.
 
\begin{theorem}\label{main-thm}
Assume that $A$ and $V$ fulfill Hypothesis~\ref{hyp-AV-smooth}
with $\gamma\in(0,\gamma_{\mathrm{c}}]$.
Let $\phi_\lambda$ be an eigenvector of $\B{A,V}$
corresponding to an eigenvalue 
$\lambda<1$ and $\triangle$ be the function defined in \eqref{def-triangle}.
Then $\phi_\lambda\in C^\infty(\RR^3\setminus\{0\},\CC^4)$
and,
for all $a<\triangle(\lambda)$ and 
$\beta\in\NN_0^3$,
we find some $C(\lambda,a,\beta)\in(0,\infty)$
such that
$$
\forall\;x\in\RR^3\,,\;|x|\grg1\::\quad
|\partial^\beta_x\phi_\lambda(x)|\,\klg\,C(\lambda,a,\beta)\,e^{-a|x|}\,.
$$
\end{theorem}

\begin{proof}
The statement follows from Theorem~\ref{thm-exp-Hk} and the
Sobolev embedding theorem.
\end{proof}


\subsection{The Dirac operator}

\noindent
As a remark we state $L^2$- and pointwise exponential decay
estimates for the Dirac operator although they are consequences of
the $L^2$-estimates in \cite{BeGe} and the ellipitic regularity
of $\D{0}$. In Section~\ref{sec-exp-decay} we present a proof
of the $L^2$-exponential localization of spectral projections
of the Dirac operator since our argument
-- a new variant of one given in \cite{BFS1} which easily
extends to the no-pair operator -- is particularly simple
in this case.

We assume that  $V$ fulfills Hypothesis~\ref{hyp-Vogel}
and that $A\in L_\loc^\infty(\RR^3,\RR^3)$ in what follows.
Then it is well-known
(and explained in more detail in \cite[Proposition~4.3]{RTdA})
that the results of \cite{BdMP,Ch,Ne1} 
ensure the existence 
of a distinguished self-adjoint extension, $\DAV$, of the 
Dirac operator defined by
$$
\DAV\,\vp\,:=\,(\D{0}+\alpha\cdot A+V)\,\vp\,,\qquad \vp\in\COI\,.
$$ 
This extension is uniquely determined by the conditions
\begin{enumerate}
\item[(i)] 
$\dom(\D{A,V})\subset 
H_\loc^{1/2}(\RR^3,\CC^4)$.
\item[(ii)] For all 
$\psi\in H^{1/2}(\RR^3,\CC^4)$ having compact support and all 
$\phi\in\dom(\D{A,V})$,
$$
\SPn{\psi}{\D{A,V}\,\phi}
\,=\,
\SPb{|\D{0}|^{1/2}\,\psi}{\sgn(\D{0})\,|\D{0}|^{1/2}\,\phi}\,+\,
\SPb{|X|^{1/2}\,\psi}{U\,
|X|^{1/2}\,\phi}
\,,
$$
where $U\,|X|$ is the polar decomposition
of $X:=\alpha\cdot A+V$.
\end{enumerate}
Standard arguments show that 
$\DAV$ has the local compactness property and since $V$ drops off
to zero at infinity this in turn implies that
\begin{equation}\label{specess-DAV}
\specess(\DAV)\,=\,\specess(\DA)\,\subset\,
(-\infty,-1]\cup[1,\infty)\,.
\end{equation}

\begin{theorem}\label{main-thm-DAV}
Assume that $A\in L_\loc^\infty(\RR^3,\RR^3)$
and that $V$ fulfills Hypothesis~\ref{hyp-Vogel}.
Let $\phi_\lambda$ be a normalized 
eigenvector of $\DAV$ corresponding to an
eigenvalue $\lambda\in(-1,1)$. Then, for all
$a\in(0,\sqrt{1-\lambda^2})$, there is some 
$A$-independent constant $C(\lambda,a)\in(0,\infty)$
such that
$$
\big\|\,e^{a|\cdot|}\,\phi_\lambda\,\big\|\,\klg\,C(\lambda,a)\,.
$$
Assume additionally
that $A$ and $V$ fulfill Hypothesis~\ref{hyp-AV-smooth}.
Then we have
$\phi_\lambda\in C^\infty(\RR^3\setminus\{0\},\CC^4)$
and, for all
$a\in(0,\sqrt{1-\lambda^2})$ and $\beta\in\NN_0^3$,
there is some $C(\lambda,a,\beta)\in (0,\infty)$ such that
\begin{equation}\label{exp-bd-Dirac}
\forall\;x\in\RR^3\,,\;|x|\grg1\::\qquad
\big|\,\partial_x^\beta\phi_\lambda(x)\,\big|
\,\klg\,C(\lambda,a,\beta)\,e^{-a|x|}\,.
\end{equation}
\end{theorem}

\begin{proof}
The assertions follow from Theorem~\ref{thm-L2-exp-Dirac}
and standard arguments using the elliptic regularity
of $\D{0}$ and the Sobolev embedding theorem. (Terms containing 
derivatives of the vector potential $A$
are dealt with as in \eqref{fedi1}.)
\end{proof}


\subsection{Examples}

To complete the picture we state some conditions on $A$ and $V$
which ensure the existence of infinitely many eigenvalues
of $\BAVC$ (resp. $\DAV$) below $1$ (resp. in $(-1,1)$)
and which imply that the essential spectrum covers the whole
half-line $[1,\infty)$ (resp. $(-\infty,-1]\cup[1,\infty)$).
The properties of $A$ which are explicitly used
in the proofs are stated in the following hypothesis, where
$$
\ball{R}{y}\,:=\,\big\{\,x\in\RR^3\::\;|x-y|<R\,\big\}\,,\qquad y\in\RR^3\,,
\;R>0\,.
$$

\begin{hypothesis}\label{hyp-A-ex}
(i) $A\in C^\infty(\RR^3,\RR^3)$
and, for every $\lambda\grg1$,
there exist radii, $1\klg R_1<R_2<\ldots$, $R_n\nearrow\infty$,
and normalized spinors
$\psi_1(\lambda),\psi_2(\lambda),\ldots\in\COI$
such that 
\begin{equation}\label{Weyl-DA}
\supp(\psi_n(\lambda))\subset\RR^3\setminus\ball{R_n}{0}\,,
\quad\lim_{n\to\infty}
(\DA-\lambda)\,\psi_n(\lambda)\,=\,0\,.
\end{equation}
(ii) The assumptions of Part~(i) are fulfilled
and the Weyl sequence $\{\psi_n(1)\}_{n\in\NN}$
has the following additional properties:
Its elements have vanishing lower spinor components,
$\psi_n(1)=(\psi_{n,1}(1),\psi_{n,2}(1),0,0)^\top$, $n\in\NN$,
there is some $\delta\in(0,1)$ such that
\begin{equation}\label{supp-psin-ex}
\supp(\psi_n(1))\,\subset\,\big\{\,R_n<|x|<(1+\delta)\,R_n\,\big\}\,,
\qquad 2R_n\,\klg\,R_{n+1}\,,
\end{equation}
for all $n\in\NN$, and 
\begin{equation}\label{Weyl-DA-ex}
\big\|\,\big(\DA-1\big)\,\psi_n(1)\,\big\|\,=\,\bigO(1/R_n)\,,
\qquad n\to\infty\,.
\end{equation}
\end{hypothesis}

\bigskip

\noindent
Obviously, the vectors $\psi_n(\lambda)$ in \eqref{Weyl-DA}
form a Weyl sequence for $\DAV$ and it is easy to see
that their projections onto $\HR^+_A$ define a Weyl sequence for 
$\B{A,V}$. Under Hypothesis~\ref{hyp-A-ex}(ii) the vectors
$\psi_n(1)$ can be used as test functions in a minimax principle
to prove the existence of infinitely many bound states.

To give some explicit conditions
we recall a result from \cite{HNW} 
which provides a large class
of examples where
Hypotheses~\ref{hyp-A-ex}(i)\&(ii) are fulfilled

\begin{example}[\cite{HNW}]\label{ex1-HNW}
{\em
(i)
Suppose that $A\in C^\infty(\RR^3,\RR^3)$, $B=\rot \,A$,
and set, for $x\in\RR^3$ and $\nu\in\NN$,
\begin{equation*}\label{def-epsnu}
\epsilon_0(x)\,:=\,|B(x)|\,,\qquad
\epsilon_\nu(x)\,:=\,
\frac{\sum_{|\alpha|=\nu} |\partial^\alpha B(x)|}{
1+\sum_{|\alpha|<\nu} |\partial^\alpha B(x)|}\,.
\end{equation*}
Suppose further that
there exist $\nu\in\NN_0$, $z_1,z_2,\ldots\in\RR^3$, and 
$\rho_1,\rho_2,\ldots>0$ such that $\rho_n\nearrow\infty$,
the balls $\ball{\rho_n}{z_n}$, $n\in\NN$, are mutually
disjoint and
\begin{equation*}\label{hyp-Weyl-eq2}
\sup\big\{\,\epsilon_\nu(x)\,\big|\:x\in\ball{\rho_n}{z_n}\,\big\}
\,\longrightarrow\,0\,,\qquad n\to\infty\,.
\end{equation*}
Then $A$ fulfills Hypothesis~\ref{hyp-A-ex}(i).
This follows directly from the constructions presented in \cite{HNW}.

\smallskip

\noindent
(ii) Suppose additionally that there
is some $C\in(0,\infty)$ such that
$\rho_n<|z_n|\klg C\,\rho_n$, for all $n\in\NN$,
and that either
$$
\sup\big\{\,|B(x)|\,:\;x\in\ball{\rho_n}{z_n}\,\big\}\,\klg\,
C/|z_n|^2\,,\qquad n\in\NN\,,
$$
or
\begin{equation*}\label{hyp-HNW-ex}
\forall\;n\in\NN\::\;\;|B(z_n)|\grg1/C\quad
\textrm{and}\quad
\sup\big\{\,\epsilon_\nu(x)\,\big|\:x\in\ball{\rho_n}{z_n}\,\big\}
\,=\,o(\rho_n^{-\nu})\,.
\end{equation*}
Then $A$ fulfills Hypothesis~\ref{hyp-A-ex}(ii).
This follows by inspecting and adapting the relevant proofs
in \cite{HNW}. Since this procedure is straight-forward but
a little bit lengthy we refrain from explaining any detail here.
\hfill$\Box$
}
\end{example}

\begin{theorem}\label{thm-ex-no-pair}
Assume that $V$ fulfills Hypothesis~\ref{hyp-Vogel}
with $\gamma\in(0,\gamma_{\mathrm{c}}]$.
If $A$ fulfills Hypothesis~\ref{hyp-A-ex}(i) then
$\specess(\B{A,V})=[1,\infty)$.
If $A$ fulfills Hypothesis~\ref{hyp-A-ex}(ii) and if
\begin{equation}\label{V=O(1/R)}
\exists\;\wt{\gamma}>0\quad\forall\;x\in\RR^3\setminus\{0\}\::
\qquad \max_{v\in\CC^4:\,|v|=1}\SPn{v}{V(x)\,v}\,\klg\,
-\wt{\gamma}\,\min\{1,|x|^{-1}\}\,,
\end{equation} 
then $\B{A,V}$
has infinitely many eigenvalues below and
accumulating at $1$.
\end{theorem}

\begin{proof}
Theorem~\ref{thm-ex-no-pair} follows from Theorem~\ref{main-thm-L2},
Lemma~\ref{thm-upper-bd},
and Theorem~\ref{thm-eigenvalues}.
\end{proof}

\begin{theorem}\label{thm-ex-Dirac}
Assume that $V$ fulfills Hypothesis~\ref{hyp-Vogel}.
If $A$ fulfills Hypothesis~\ref{hyp-A-ex}(i),
then $\specess(\DAV)=(-\infty,-1]\cup[1,\infty)$.
If $A$ fulfills Hypothesis~\ref{hyp-A-ex}(ii) and $V$ satisfies
\eqref{V=O(1/R)}
then there exist infinitely many eigenvalues of $\DAV$ in $(-1,1)$.
\end{theorem}

\begin{proof}
In view of \eqref{specess-DAV} and
since $V$ drops off to zero at infinity
the first statement is clear. The second assertion
is a special case of
\cite[Theorem~2.9]{MaSt08a}.
\end{proof}


\section{Miscellaneous results on spectral projections}
\label{sec-spectral-proj}

\noindent
In order to obtain any of our results on the no-pair operator
it is crucial from a technical point of view to have some
control on commutators of $\PA$ with multiplication operators
and on the difference between $\PA$ and $\PO$.
Appropriate estimates are derived in this section. They
are based on the formula \eqref{def-PA} and
the representation of the
sign function of a self-adjoint operator, $T$, acting in some
Hilbert space, $\mathscr{K}$, with $0\in\vr(T)$
as a strongly
convergent Cauchy principal value, 
\begin{equation}\label{Cauchy-PV}
\sgn(T)\,\psi\,=\,
T\,|T|^{-1}\,\psi
\,=\,
\lim_{\tau\to\infty}\int_{-\tau}^\tau(T-iy)^{-1}\,\psi\,\frac{dy}{\pi}\,,
\qquad \psi\in\mathscr{K}\,.
\end{equation}
We write
\begin{equation}\label{def-RAV}
\RAV{z}\,:=\,(\DAV-z)^{-1}\,,\qquad
\RA{z}\,:=\,(\DA-z)^{-1}\,,
\end{equation}
in what follows.
Then another frequently used identity is
\begin{eqnarray}
\lefteqn{
\label{res-id}
\R{\wt{A},\wt{V}}{z}\,\mu\,-\,\mu\,\RAV{z}
}
\\
&=&\nonumber
\R{\wt{A},\wt{V}}{z}\,\big(i\alpha\cdot\nabla\mu-\mu\,(V-\wt{V})
+\mu\,\alpha\cdot(A-\wt{A})\big)\,\RAV{z}\,,
\end{eqnarray}
where $z\in\vr(\DAV)\cap\vr(\D{\wt{A},\wt{V}})$.
Here we assume that
$V$ and $\wt{V}$ fulfill Hypothesis~\ref{hyp-Vogel},
such that $\wt{V}$ and $\mu\,(V-\wt{V})$ are bounded, matrix-valued
multiplication operators,
and that $A$, $\wt{A}$, and $\mu$ satisfy
\begin{equation}\label{hyp-AwtAmu}
\left\{
\begin{array}{l}
A,\wt{A}\in L_\loc^\infty(\RR^3,\RR^3)\,,\quad \mu\in C^\infty(\RR^3,\RR)\,,
\\
\|\mu\|_\infty+\|\nabla\mu\|_\infty+\|\mu\,(A-\wt{A})\|_\infty\,<\,\infty\,.
\end{array}
\right.
\end{equation}
(Using the essential self-adjointness of
$\D{\wt{A},\wt{V}}\!\!\upharpoonright_{\COI}$, 
it is actually simpler and sufficient
to derive the adjoint of \eqref{res-id}.)
A combination of \eqref{def-PA}, \eqref{Cauchy-PV}, and \eqref{res-id} 
yields the
following formula, where
$\phi,\psi\in\HR$ and $A,\wt{A},\mu$ satisfy \eqref{hyp-AwtAmu},
\begin{eqnarray}
\lefteqn{
\label{basic-for-comm}
\SPb{\phi}{(\Lambda^+_{\wt{A}}\,\mu-\mu\,\PA)\,\psi}
}
\\
&=&
\lim_{\tau\to\infty}
\int_{-\tau}^\tau\nonumber
\SPb{\phi}{\R{\wt{A}}{iy}\,\alpha\cdot\big(i\nabla\mu\,
+\,\mu\,(A-\wt{A})\big)\,\RA{iy}
\,\psi}\,\frac{dy}{2\pi}\,.
\end{eqnarray}
We also recall the identities
\begin{equation}\label{Cauchy-Trick}
\int_\RR\big\|\,|\DA|^{1/2}\,\RA{iy}\,\psi\,\big\|^2\,dy
\,=\,
\int_\RR\int_\RR\frac{|\lambda|}{\lambda^2+y^2}\,dy\,
d\|E_\lambda(\DA)\psi\|^2\,=\,\pi\,\|\psi\|^2,
\end{equation}
for all $\psi\in\HR$, $A\in L_\loc^\infty(\RR^3,\RR^3)$, and 
\begin{eqnarray}
\|\RA{iy}\|&=&(1+y^2)^{-1/2}\,,\qquad y\in\RR\,,\;A\in 
L_\loc^\infty(\RR^3,\RR^3)\,,\label{RA-norm}
\\
\|\,\alpha\cdot v\,\|_{\LO(\CC^4)}&=&|v|\,,\qquad\qquad \qquad v\in\RR^3
\label{alpha-v-norm}\,.
\end{eqnarray}
Here \eqref{alpha-v-norm} follows from \eqref{Clifford}.
Finally, we need the following crucial estimate stating that
$\RA{z}$ stays bounded after
conjugation with suitable exponential weights, $e^F$,
acting as multiplication operators in $\HR$.
Although it is well-known
(see, e.g., \cite{BeGe}), we recall its proof since it determines
the exponential decay rates in our main theorems.

\begin{lemma}\label{le-marah}
Let $A\in L^\infty_\loc(\RR^3,\RR^3)$,
$\lambda\in(-1,1)$, $y\in\RR$, $a\in[0,\sqrt{1-\lambda^2})$,
and let
$F\in C^\infty(\RR^3,\RR)$ have a fixed sign and
satisfy $|\nabla F|\klg a$.
Then 
$\lambda+iy\in\vr(\DA+i\alpha\cdot\nabla F)$,
\begin{equation}\label{marah0}
e^F\,\RA{\lambda+iy}\,e^{-F}
=(\DA+i\alpha\cdot \nabla F+\lambda
+iy)^{-1}\!\!\upharpoonright_{\dom(e^{-F})}\,,
\end{equation}
and
\begin{equation}\label{marah1}
\big\|\,e^F\,\RA{\lambda+iy}\,e^{-F}\,\big\|\,
\klg\,\frac{\sqrt{3}\sqrt{1+y^2+\lambda^2+a^2}}{
1+y^2-\lambda^2-a^2}\,\klg\,
\frac{\sqrt{3}}{\sqrt{1+y^2}}
\frac{\sqrt{1+\lambda^2+a^2}}{1-\lambda^2-a^2}
\,.
\end{equation}
\end{lemma}

\begin{proof}
A straightforward calculation 
yields, for $z=\lambda+iy$, $\lambda\in(-1,1)$, $y\in\RR$, 
$\ve>0$, and $\vp\in\COI$,
\begin{eqnarray*}
\lefteqn{
\frac{1}{4\ve}\:\big\|\,e^F\,(\DA-z)\,e^{-F}\,\vp\,\big\|^2
\,+\,3\ve\,\big\|\,\alpha\cdot(-i\nabla+A)\,\vp\,\big\|^2
}
\\
& &
\,+\,3\ve\,(1+|z|^2)\,\|\vp\|^2\,+\,3\ve\,\SPb{\vp}{|\nabla F|^2\,\vp}
\\
& &\quad\grg\,
\Re\,\SPb{e^{-F}\,(\DA+\ol{z})\,e^{F}\,\vp}{e^F\,(\DA-z)\,e^{-F}\,\vp}
\\
& &\quad=\,
\big\|\,\alpha\cdot(-i\nabla+A)\,\vp\,\big\|^2
\,+\,\SPb{\vp}{\big(1-\Re z^2-|\nabla F|^2\big)\,\vp}\,.
\end{eqnarray*}
Together with $|\nabla F|\klg a$ and
$\Re z^2=\lambda^2-y^2$ this implies
$$
\big\|\,e^F\,(\DA-z)\,e^{-F}\,\vp\,\big\|^2
\,\grg\,
4\ve\,(b_--\,3\ve\,b_+)
\,\|\vp\|^2\,,
$$
where $b_\pm:=1+y^2\pm\lambda^2\pm a^2$.
The optimal choice for $\ve$ is
$\ve=b_-/(6b_+)$.
Since $\COI$ is a core for the closed operator 
$\DA+i\alpha\cdot\nabla F$ with domain $\dom(\DA)$
it follows that
\begin{equation}\label{marah99}
\big\|\,(\DA+i\alpha\cdot\nabla F-z)\,\psi\,\big\|
\,\grg\,b_-\,(3b_+)^{-1/2}
\,\|\psi\|\,,\qquad \psi\in\dom(\DA)\,.
\end{equation}
We may replace $F,z$ by $-F,\ol{z}$ in \eqref{marah99},
whence $\Ran(\DA+i\alpha\cdot\nabla F-z)^\bot=
\Ker(\DA-i\alpha\cdot\nabla F-\ol{z})=\{0\}$.
On the other hand we know that $\Ran(\DA+i\alpha\cdot\nabla F-z)$
is closed since $\DA+i\alpha\cdot\nabla F-z$ is closed with a
continuous inverse.
It follows that $z\in\vr(\DA+i\alpha\cdot\nabla F)$.

We assume that $F\grg0$ in the rest of this proof. Let $\psi\in\HR$.
We pick a sequence, $\{\vp_n\}_{n\in\NN}\in\COI^\NN$,
which converges to 
$\eta:=(\DA+i\alpha\cdot\nabla F-z)^{-1}\,\psi\in\dom(\DA)$
with respect to the graph norm of $\DA$.
Passing to the limit in
$$
\RA{z}\,e^{-F}\,(\DA+i\alpha\cdot \nabla F-z)\,\vp_n
\,=\,e^{-F}\,\vp_n\,,
$$
we obtain
$\RA{z}\,e^{-F}\,\psi=e^{-F}\,\eta$, which implies
\begin{equation}\label{mildred11}
e^F\,\RA{z}\,e^{-F}\,=\,(\DA+i\alpha\cdot\nabla F-z)^{-1}\,.
\end{equation}
Taking the adjoint we get
\begin{equation}\label{mildred22}
e^{-F}\,\RA{\ol{z}}\,e^F\,\subset\,
\big(e^F\,\RA{z}\,e^{-F}\big)^*\,=\,(\DA-i\alpha\cdot\nabla F-\ol{z})^{-1}
\,.
\end{equation}
\eqref{mildred11} and \eqref{mildred22} together prove \eqref{marah0}.
\end{proof}

\bigskip

\noindent
To shorten the presentation and 
since it is sufficient for our applications below 
we consider only bounded weight functions $F$ in the following
Lemma~\ref{le-peki}. 
Similar estimates have already been derived in 
\cite{MaSt08a}.

\begin{lemma}\label{le-peki}
Let
$A\in L^\infty_\loc(\RR^3,\RR^3)$, $a\in[0,1)$,
$\chi\in C^\infty(\RR^3,[0,1])$
with $\nabla\chi\in C_0^\infty(\RR^3,\RR^3)$, and let
$F\in C^\infty(\RR^3,\RR)\cap L^\infty(\RR^3,\RR)$ have a fixed sign and
satisfy $|\nabla F|\klg a$.
Then
\begin{equation}\label{peki}
\big\|\,|\DA|^{1/2}\,[\PA\,,\,\chi\,e^F]\,e^{-F}\,\big\|\,\klg
\,\frac{\sqrt{6}}{2}\cdot\frac{a+\|\nabla\chi\|_\infty}{1-a^2}\,.
\end{equation}
In particular,
\begin{equation}\label{peki2}
\big\|\,\PAF\,\big\|\,\klg\,1\,
+\,\frac{\sqrt{6}}{2}\cdot\frac{a}{1-a^2}\,,\qquad
\textrm{where}\quad \PAF\,:=\,e^F\,\PA\,e^{-F}\,.
\end{equation}
Moreover,
\begin{equation}\label{sonja99}
\big\|\,
e^F\,[\chi\,,\,\PA]
\,\big\|\,\klg\,\frac{\sqrt{6}}{2(1-a^2)}\,\|\,e^F\,\nabla\chi\,\|_\infty\,.
\end{equation}
(ii)
Assume additionally that $\nabla\chi=\nabla F=0$ 
in a neighbourhood, $\cU\subset\RR^3$, of $0$
and let $\zeta\in C_0^\infty(\cU,[0,1])$.
Then
\begin{equation}\label{Vpeki}
\Big\|\,\frac{\zeta}{|\cdot|}
\,[\PA\,,\,\chi\,e^F]\,e^{-F}\,\Big\|\,\klg\,\sqrt{6}
\,\big(\|\nabla\zeta\|_\infty+\|\zeta\,A\|_\infty\big)\,
\frac{a+\|\nabla\chi\|_\infty}{1-a^2}
\,.
\end{equation}
If $a=0$, then the factor $\sqrt{6}$ in \eqref{peki},
\eqref{sonja99}, and \eqref{Vpeki} can be replaced by $1$.
\end{lemma}

\begin{proof}
On account of \eqref{basic-for-comm}
we have, for $\phi\in\dom(|\DA|^{1/2})$ and $\psi\in\HR$,
\begin{eqnarray}
\lefteqn{\nonumber
\big|\SPb{|\DA|^{1/2}\,\phi}{[\PA\,,\,\chi\,e^F]\,e^{-F}\,\psi}\big|
}
\\
&\klg&
\int_{\RR}\Big|
\nonumber
\SPB{|\DA|^{1/2}\,\phi}{\RA{iy}\,i\alpha\cdot(\nabla\chi+\chi\,\nabla F)
\,e^F\,\RA{iy}\,e^{-F}\,\psi}\Big|\,\frac{dy}{2\pi}
\\
&\klg&\nonumber
\sqrt{6}\:\frac{a+\|\nabla\chi\|_\infty}{2\pi(1-a^2)}\,
\Big(\int_\RR\big\|\,|\DA|^{1/2}\,\RA{-iy}\,\phi\,\big\|^2\,dy\Big)^{1/2}
\Big(\int_\RR\frac{\|\psi\|^2\,dy}{1+y^2}\Big)^{1/2}.
\end{eqnarray} 
In the last line we have used 
\eqref{alpha-v-norm} and \eqref{marah1}
(with $\lambda=0$ and $\sqrt{3}\sqrt{1+a^2}\klg\sqrt{6}$).
Applying \eqref{Cauchy-Trick}
we conclude
that 
$[\PA\,,\,\chi\,e^F]\,e^{-F}\,\psi\in\dom(|\DA|^{1/2*})=\dom(|\DA|^{1/2})$
and that \eqref{peki} holds true.
The bound \eqref{sonja99}
follows from
$$
\big|\SPb{e^F\,\phi}{[\chi\,,\,\PA]\,\psi}\big|
\,\klg\,
\int_{\RR}\Big|\SPb{e^{-F}\,\R{A}{-iy}e^F\,\phi}{
(e^F\,\alpha\cdot\nabla\chi)\,\R{A}{iy}\,\psi}\Big|\,\frac{dy}{2\pi}\,,
$$
for all $\phi,\psi\in\HR$, together with
\eqref{RA-norm}, \eqref{alpha-v-norm}, and \eqref{marah1}.

In order to prove Part~(ii) we first observe that the additional
assumption implies $\zeta\,(\nabla\chi+\chi\,\nabla F)=0$.
Together with \eqref{RA-norm} and \eqref{res-id} (with $\wt{A}=0$)
this permits 
to get,
for $\phi\in H^1(\RR^3,\CC^4)$ and $\psi\in\HR$,
\begin{eqnarray*}
\lefteqn{
\big|\SPb{\zeta\,|\cdot|^{-1}\,\phi}{[\PA\,,\,\chi\,e^F]\,e^{-F}\,\psi}\big|
}
\\
&\klg&
\int_{\RR}\Big|\nonumber
\SPB{\frac{1}{|\cdot|}
\,\phi}{\zeta\,\RA{iy}\,i\alpha\cdot(\nabla\chi+\chi\,\nabla F)
\,e^F\,\RA{iy}\,e^{-F}\,\psi}\Big|\,\frac{dy}{2\pi}
\\
&=&
\int_{\RR}\Big|\nonumber
\SPB{\frac{1}{|\cdot|}\,\phi}{
\R{0}{iy}\,\alpha\cdot\{i\nabla\zeta+\zeta\,A\}\,\times
\\
& &\qquad\times
\,\RA{iy}\,i\alpha\cdot(\nabla\chi+\chi\,\nabla F)
\,e^F\,\RA{iy}\,e^{-F}\,\psi}\Big|\,\frac{dy}{2\pi}
\\
&\klg&\frac{\sqrt{6}}{2\pi}
\int_\RR
\big\|\,\R{0}{-iy}\,|\cdot|^{-1}\,\big\|
\,\big(\|\nabla\zeta\|_\infty+\|\zeta\,A\|_\infty\big)
\frac{a+\|\nabla\chi\|_\infty}{1-a^2}\,
\frac{\|\phi\|\,\|\psi\|\,dy}{1+y^2}\,.
\end{eqnarray*}
By Hardy's inequality 
$\|\R{0}{-iy}\,|\cdot|^{-1}\|
=\|\,|\cdot|^{-1}\,\R{0}{iy}\|\klg 2$, for all
$y\in \RR$. 

The last statement of this lemma follows from an obvious
modification of
the proof above. In fact, in the case $a=0$ we can always
use \eqref{RA-norm} where \eqref{marah1} has been applied before.
\end{proof}

\bigskip

\noindent
In what follows we set, for any vector-valued function
$u:\RR^3\to\CC^3$,
$$
\|u\|_{\Lip^1}\,:=\,\sum_{i=1}^3\Big\{\,
\|u_i\|_\infty\,+\,\sup_{x\not=y}\,\frac{|u_i(x)-u_i(y)|}{|x-y|}\,\Big\}\,.
$$

\begin{lemma}\label{le-vgl-PA-PO}
Let $A\in L^\infty_\loc(\RR^3,\RR^3)$ and
$\mu\in C_0^\infty(\RR^3,\CC)$. Then
\begin{eqnarray}
\big\|\,|\D{0}|^{1/2}\,(\PO\,\mu-\mu\,\PA)\,|\DA|^{1/2}\,\big\|
&\klg&\frac{1}{2}\,\big(\|\nabla\mu\|_\infty+\|\mu\,A\|_\infty\big)
\label{rudi0}
\,.
\end{eqnarray}
Assume further that $A$ is Lipschitz continuous
in a neighbourhood, $\cU\subset\RR^3$, of 
$\supp(\mu)$ and let $\chi\in C_0^\infty(\cU,[0,1])$
be such that $\chi\,\mu=\mu$. Then
\begin{equation}
\big\|\,\D{0}\,(\PO\,\mu-\mu\,\PA)\,\big\|
\,\klg\,\label{rudi2a}
\frac{1}{2}\,\big(\|\nabla\mu\|_{\Lip^1}+\|\mu\,A\|_{\Lip^1}\big)
\big(\sqrt{3}+\|\nabla\chi\|_\infty+\|\chi\,A\|_\infty\big)
\,.
\end{equation}
\end{lemma}

\begin{proof}
In view of \eqref{basic-for-comm} and \eqref{Cauchy-Trick}
we have, for $\vp\in\dom(|\D{0}|^{1/2})$ and $\psi\in\dom(|\DA|^{1/2})$,
\begin{eqnarray*}
\lefteqn{
\big|\SPb{|\D{0}|^{1/2}\,\vp}{(\PO\,\mu-\mu\,\PA)\,|\DA|^{1/2}\,\psi}\big|
}
\\
&\klg&
\int_\RR
\Big|\SPB{|\D{0}|^{1/2}\,\vp}{
\R{0}{iy}\,\alpha\cdot(i\nabla\mu+\mu\,A)\,\RA{iy}\,
|\DA|^{1/2}\,\psi}\Big|\,\frac{dy}{2\pi}
\\
&\klg&
\frac{\|\nabla\mu\|_{\infty}+\|\mu\,A\|_\infty}{2}\:\|\vp\|\,\|\psi\|\,.
\end{eqnarray*}
This implies \eqref{rudi0}.
In order to prove \eqref{rudi2a} we use 
$i\nabla\mu+\mu\,A=(i\nabla\mu+\mu\,A)\,\chi$ and \eqref{res-id}
to write
\begin{eqnarray*}
\lefteqn{
\R{0}{iy}\,\alpha\cdot(i\nabla\mu+\mu\,A)\,\RA{iy}
}
\\
&=&
\R{0}{iy}\,\alpha\cdot(i\nabla\mu+\mu\,A)\,\R{0}{iy}\,\chi
\\
& &\;
-\;
\R{0}{iy}\,\alpha\cdot(i\nabla\mu+\mu\,A)\,\R{0}{iy}\,
\alpha\cdot(i\nabla\chi+\chi\,A)\,\RA{iy}\,.
\end{eqnarray*}
This identity yields, for all $\vp\in\dom(\D{0})$
and $\psi\in\HR$,
\begin{eqnarray}
\lefteqn{
\big|\SPb{\D{0}\,\vp}{(\PO\,\mu-\mu\,\PA)\,\psi}\big|\nonumber
}
\\
&\klg&
\int_\RR\big\|\,|\D{0}|^{1/2}\,\R{0}{-iy}\,\vp\,\big\|
\,\big\|\,|\D{0}|^{1/2}\,\alpha\cdot(i\nabla\mu+ \mu\,A)
\,|\D{0}|^{-1/2}\,\big\|\nonumber
\\
& &\qquad\;\;\cdot\,
\big\|\,|\D{0}|^{1/2}\,\R{0}{iy}\,\psi\,\big\|
\,\frac{dy}{2\pi}\nonumber
\\
& &+
\int_\RR\big\|\,\D{0}\,\R{0}{-iy}\,\big\|
\,\big(\|\nabla\mu\|+\|\mu\,A\|\big)\,
\,\big(\|\nabla\chi\|+\|\chi\,A\|\big)\,
\frac{\|\vp\|\,\|\psi\|\,dy}{2\pi(1+y^2)}\,.\label{yvonne}
\end{eqnarray}
Since each matrix entry of $M:=\alpha\cdot (i\nabla\mu+\mu\,A)$ 
is a Lipschitz
continuous,
compactly supported function 
and since $\alpha_i$ comutes with $|\D{0}|^{1/2}$ we 
readily verify (e.g., by using an explicite integral
formula for $\|(1-\Delta)^{1/4}\,f\|^2$ \cite[Theorem~7.12]{LL-Buch}; 
see Appendix~\ref{app-K2})
that
\begin{eqnarray*}
\big\|\,|\D{0}|^{1/2}\,M\,|\D{0}|^{-1/2}\,\big\|
&\klg&
\sum_{i=1}^3\big\|\,|\D{0}|^{1/2}\,(\partial_i\mu+\mu\,A_i)
\,|\D{0}|^{-1/2}\,\big\|
\\
&\klg&
\sqrt{3}\,\big(\|\nabla\mu\|_{\Lip^1}+\|\mu\,A\|_{\Lip^1}\big)\,.
\end{eqnarray*}
Therefore,
\eqref{rudi2a} follows
from the above estimates and \eqref{Cauchy-Trick}.
\end{proof}

\begin{remark}
{\em
It can easily be read off from the previous proof that
\begin{equation}\label{Sergey-est}
\big\|\,\D{0}\,[\PO\,,\,\mu]\,\big\|\,\klg\,
\frac{\sqrt{3}}{2}\:\|\nabla\mu\|_{\Lip^1}\,.
\end{equation}
(In fact, if $A=0$ then the term in \eqref{yvonne} is superfluous.)
A similar bound 
has been derived in \cite{MoVu}
by means of an explicite formula for the integral kernel
of~$\PO$.
}
\end{remark}

\begin{lemma}
Let
$A\in L^\infty_\loc(\RR^3,\RR^3)$, 
$\chi\in C^\infty(\RR^3,[0,1])$ with 
$\nabla\chi\in C_0^\infty(\RR^3,\RR^3)$, and let
$\wt{\chi}\in C_0^\infty(\RR^3,[0,1])$
satisfy $\wt{\chi}\equiv1$ on $\supp(\nabla\chi)$.
Then
\begin{equation}\label{DA-commutator}
\big\|\,\DA\,[\PA\,,\,\chi]\,\big\|\,\klg\,\frac{3}{2}\,\|\chi\|_{\Lip^1}
\,\big(1+
\|\nabla\wt{\chi}\|_\infty+\|\wt{\chi}\,A\|_\infty\big)\,.
\end{equation}
\end{lemma}

\begin{proof}
Using $\nabla\chi=\wt{\chi}\,\nabla\chi$ and \eqref{res-id},
we write the term appearing on the right side of \eqref{basic-for-comm}
as
\begin{eqnarray*}
\RA{z}\,i\alpha\cdot\nabla\chi\,\RA{z}
&=&\nonumber
\wt{\chi}\,\R{0}{z}\,i\alpha\cdot(\nabla\chi)\,\R{0}{z}\,\wt{\chi}
\\
& &
\;-\,
\RA{z}\,i\alpha\cdot\nabla\chi\,\R{0}{z}\,\alpha\cdot
(i\nabla\wt{\chi}+\wt{\chi}\,A)\,\RA{z}
\\
& &
\;+\,\RA{z}\,\alpha\cdot(i\nabla\wt{\chi}-\wt{\chi}\,A)\,
\R{0}{z}\,i\alpha\cdot\nabla\chi\,\R{0}{z}\,\wt{\chi}
\,.
\end{eqnarray*}
Using this we infer from \eqref{basic-for-comm} (with $\phi=\DA\,\vp$)
that, for $\vp\in\COI$ and $\psi\in\HR$,
\begin{eqnarray*}
\lefteqn{\nonumber
\big|\SPb{\DA\,\vp}{[\PA\,,\,\chi]\,\psi}\big|
}
\\
&\klg&
\Big|\lim_{\tau\to\infty}\int_{-\tau}^\tau 
\SPb{\wt{\chi}\,(\D{0}+\alpha\cdot A)\,\vp}{ 
\R{0}{z}\,i\alpha\cdot(\nabla\chi)\,\R{0}{z}\,\wt{\chi}
\,\psi}\,
\frac{dy}{2\pi}\,\Big|
\\
& &\;+\;
\int_{\RR}\nonumber
\big\|\DA\,\RA{-iy}\big\|\,\|\vp\|\,
\|\nabla\chi\|_\infty\,\big(
\|\nabla\wt{\chi}\|_\infty+\|\wt{\chi}\,A\|_\infty\big)
\|\psi\|
\,\frac{dy}{\pi(1+y^2)}\,.
\end{eqnarray*}
Applying \eqref{basic-for-comm} backwards,
we thus obtain
\begin{eqnarray*}
\lefteqn{\nonumber
\big|\SPb{\DA\,\vp}{[\PA\,,\,\chi]\,\psi}\big|
}
\\
&\klg&
\big|\SPb{\wt{\chi}\,\vp}{\D{0}\,[\PO\,,\,\chi]\,\wt{\chi}\,\psi}\big|
\,+\,
\big|\SPb{\alpha\cdot(i\nabla\wt{\chi}+\wt{\chi}\, A)\,\vp}{
[\PO\,,\,\chi]\,\wt{\chi}\,\psi}\big|
\\
& &
\;+\;\|\nabla\chi\|_\infty\,\big(
\|\nabla\wt{\chi}\|_\infty+\|\wt{\chi}\,A\|_\infty\big)
\,\|\vp\|\,\|\psi\|\,.
\end{eqnarray*}
Taking also \eqref{Sergey-est} and $\|[\PO,\chi]\|\klg\|\nabla\chi\|/2$
into account we arrive at the assertion.
\end{proof}

\bigskip

\noindent
We close this section by stating another consequence
of the resolvent identity \eqref{res-id} showing that
the no-pair operator $\BAVC$ is actually well-defined on 
$\PA\,\COI$.

\begin{lemma}\label{VPAphi-in-L2}
Assume that 
$A\in L^\infty_\loc(\RR^3,\RR^3)$. 
Then $\PA$ maps $\dom(\DA)$ into $H^1_\loc(\RR^3,\CC^4)$.
In particular,
$V\,\PA\,\vp\in \HR$, for every $\vp\in\dom(\DA)$, provided
$V$ fulfills Hypothesis~\ref{hyp-Vogel}.
\end{lemma}

\begin{proof}
The identity \eqref{res-id} implies,
for all $\vp\in\dom(\DA)$ and $\chi\in C_0^\infty(\RR^3)$,
$$
\chi\,\PA\,\vp\,=\,\chi\,\RA{0}\,\PA\,\DA\,\vp
\,=\,
\R{0}{0}\,\big\{\chi-\alpha\cdot(i\nabla\chi+\chi\,A)\,
\RA{0}\big\}\,\PA\,\DA\,\vp
\,.
$$ 
\end{proof}


\section{Semi-boundedness of the no-pair operator}\label{sec-semi-bd}

\noindent
In the following we show that the quadratic form
of $\B{A,V}$ is bounded below on the dense subspace $\PA\,\COI\subset\HR^+_A$
provided 
one of the conditions of Theorem~\ref{thm-BRA-bb} is fulfilled.
To obtain this result we simply compare the models
with and without magnetic fields
by means of Lemma~\ref{le-vgl-PA-PO}.

\bigskip

\begin{proof}[Proof of Theorem~\ref{thm-BRA-bb}]
We pick two cutoff functions
$\mu_1,\mu_2\in C^\infty(\RR^3,[0,1])$
such that $\mu_1\equiv1$ in a neighbourhood of $0$,
$\mu_1\equiv0$ outside some larger
neighbourhood,
and $\mu_1^2+\mu_2^2=1$.
In the case $\gamma=\gamma_{\mathrm{c}}$
we may assume that $A$ is
Lipschitz continuous
on the support of $\mu_1$ by choosing the latter small enough.
In view of Hypothesis~\ref{hyp-Vogel} we may further assume
that $V\grg-\gamma/|\cdot|$. 
The following identities
are valid on $\dom(\DA)$,
\begin{eqnarray*}
\DA&=&\sum_{i=1,2}\DA\,\mu_i^2\;=\;
\sum_{i=1,2}\big\{\,\mu_i\,\DA\,\mu_i\,-\,i\alpha\cdot(\nabla\mu_i)\,
\mu_i\,\big\}
\\
&=&
\Big\{\sum_{i=1,2}\mu_i\,\DA\,\mu_i\Big\}
\,-\,i\alpha\cdot\nabla(\mu_1^2+\mu_2^2)/2
\,=\,\sum_{i=1,2}\mu_i\,\DA\,\mu_i\,.
\end{eqnarray*}
Consequently, we have, for $\vp^+\in\PA\COI$,
\begin{equation}\label{rita1}
\SPb{\vp^+}{\BAVC\,\vp^+}\,=\,
\sum_{i=1,2}\SPb{\vp^+}{\mu_i\,(\DA+V)\,\mu_i
\,\vp^+}\,.
\end{equation}
A direct application of \eqref{DA-commutator} 
yields
\begin{eqnarray}
\SPb{\vp^+}{\mu_2\,(\DA+V)\,\mu_2\,\vp^+}
&\grg&\nonumber
\SPb{\mu_2\,\vp^+}{\PA\,\DA\,\PA\,\mu_2\,\vp^+}\,-\,
\|\mu_2^2\,V\|\,\|\vp^+\|^2
\\
& &\label{friedhelm0}
\;-\,C\,\|\mu_2\|_{\Lip^1}
\,\big(1+
\|\nabla\wt{\chi}\|_\infty+\|\wt{\chi}\,A\|_\infty\big)\,\|\vp^+\|^2\,,
\end{eqnarray}
where $\wt{\chi}\in C_0^\infty(\RR^3,[0,1])$ equals one in a
neighbourhood of $\supp(\nabla\mu_2)$.
On account of Lemma~\ref{VPAphi-in-L2} we further have
$\PO\,\mu_1\,\vp^+\in\PO \,H^1(\RR^3,\CC^4)\subset
\dom(\B{0,-\gamma_{\mathrm{c}}/|\cdot|})$,
which implies
\begin{eqnarray}
\SPb{\vp^+}{\mu_1\,(\DA+V)\,\mu_1\,\vp^+}
&\grg&\nonumber
\frac{\gamma}{\gamma_{\mathrm{c}}}\,
\SPb{\mu_1\,\vp^+}{(\B{0,-\gamma_{\mathrm{c}}/|\cdot|}\oplus\POm)\,
\mu_1\,\vp^+}
\\
& &
\;+\;\label{friedhelm1a}
(1-\gamma/\gamma_{\mathrm{c}})\,\SPb{\mu_1\,\vp^+}{\PO\,\D{0}\,\PO\,
\mu_1\,\vp^+}
\\
& &
\;+\;\nonumber
\SPb{\vp^+}{\mu_1^2\,\alpha\cdot A\,\vp^+}
\,
\\
& &
\;-\;\label{friedhelm1}
2\gamma\,\Re\SPb{\mu_1\,\vp^+}{\PO\,\tfrac{1}{|\cdot|}\,\POm\,\mu_1\,\vp^+}
\\
& &
\;+\;\label{friedhelm2}
\SPb{\mu_1\,\vp^+}{\POm\,\big(\D{0}-\tfrac{\gamma}{|\cdot|}\big)
\,\POm\,\mu_1\,\vp^+}\,.
\end{eqnarray}
In order to estimate
the terms in \eqref{friedhelm1} and \eqref{friedhelm2}
we write 
$$
\POm\,\mu_1\,\vp^+\,=\,(\POm\,\mu_1-\mu_1\,\PAm)\,\vp^+
$$
and apply Lemma~\ref{le-vgl-PA-PO}.
The term in \eqref{friedhelm2} can be treated by means of \eqref{rudi0}
and Kato's inequality, $|\cdot|^{-1}\klg(\pi/2)\,|\nabla|$.
In the case $\gamma=\gamma_{\mathrm{c}}$, where $A$ is assumed to be Lipschitz
continuous on the support of $\mu_1$, the bound
\eqref{rudi2a} is available and can be applied 
together with Hardy's inequality
to estimate
the term in \eqref{friedhelm1}.
If $\gamma<\gamma_{\mathrm{c}}$ we apply \eqref{rudi0} instead and 
employ the part of the kinetic energy appearing in \eqref{friedhelm1a}
and Kato's inequality
to control the term $\ve\,\|\frac{1}{|\cdot|}\,\PO\,\mu_1\,\vp^+\|^2$
in
\begin{eqnarray*}
\lefteqn{
2\gamma\,
\big|\SPb{\mu_1\,\vp^+}{\PO\,\tfrac{1}{|\cdot|}\,\POm\,\mu_1\,\vp^+}\big|
}
\\
&\klg&
\ve\,\|\tfrac{1}{|\cdot|^{1/2}}\,\PO\,\mu_1\,\vp^+\|^2\,+\,
\frac{\gamma^2}{\ve}\,
\big\|\,\tfrac{1}{|\cdot|^{1/2}}\,(\POm\,\mu_1-\mu_1\,\PAm)\,\big\|^2
\,\|\vp^+\|^2\,,
\end{eqnarray*}
for some sufficiently small $\ve>0$.
Combining this with \eqref{friedhelm0},
we arrive at
\begin{eqnarray}
\SPn{\vp^+}{\B{A,V}\,\vp^+}\label{rita88}
&\grg&
\frac{\gamma}{\gamma_{\mathrm{c}}}\,\SPb{\mu_1\,\vp^+}{
(\B{0,-\gamma_{\mathrm{c}}/|\cdot|}\oplus\POm)
\,\mu_1\,\vp^+}\,-\,C'\,\|\vp^+\|^2\,,
\end{eqnarray}
where the constant $C'\in(0,\infty)$ does not depend on the behaviour
of $A$ outside the supports of $\wt{\chi}$ and $\mu_1$ and
certainly not on $\vp^+\in\PA\COI$.
Since $\B{0,-\gamma_{\mathrm{c}}/|\cdot|}$ is strictly positive \cite{Tix1}
this proves the theorem.
\end{proof}

\smallskip

For later reference we note that the previous
proof (recall \eqref{rita88} and the choice of $\supp(\mu_1)$) implies the 
following result:

\begin{remark}\label{rem-B0V-BAV}
{\em
If $V$ fulfills Hypothesis~\ref{hyp-Vogel} with $\gamma\in[0,\gamma_{\mathrm{c}}]$
and if
$A:\RR^3\to\RR^3$ 
is locally Lipschitz continuous then,
for every $\chi\in C_0^\infty(\RR^3,[0,1])$,
there is some $c\equiv c(\chi,A)\in(0,\infty)$ such that
$$
\SPb{\vp}{\PA\,\chi\,(
\B{0,-\gamma_{\mathrm{c}}/|\cdot|}\oplus\POm)\,\chi\,\PA\,\vp}
\,\klg\,\SPb{\vp}{\PA\,(\B{A,V}+c)\,\PA\,\vp}\,,
$$
for all $\vp\in\COI$. Since $\PA\,\COI$ is a form core of
$\B{A,V}$ this estimate implies that
$\chi\,\PA\,(\B{A,V}+c\,\PA)^{-1/2}$ maps $\HR^+_A$ into the form domain
of $\B{0,-\gamma_{\mathrm{c}}/|\cdot|}\oplus\POm$ and
\begin{equation}\label{vgl-B0V-BAV}
\big\|\,(\B{0,-\gamma_{\mathrm{c}}/|\cdot|}\oplus\POm)^{1/2}\,\chi
\,\PA\,(\B{A,V}+c\,\PA)^{-1/2}\,\big\|_{\LO(\HR^+_A,\HR)}
\,<\,\infty\,.
\end{equation}
}
\end{remark}


\section{$L^2$-exponential localization}\label{sec-exp-decay}

\noindent
In this section we derive $L^2$-exponential localization
estimates for spectral projections of the Dirac and
no-pair operators. Our proofs are new variants of an idea
from \cite{BFS1}.
We emphasize that the argument developed in \cite{BFS1}
requires no \`{a}-priori knowledge on 
the spectrum. In particular, one may first prove
the exponential localization of 
the spectral subspace corresponding to
some interval $I$ and then
infer that the spectrum in $I$ is discrete
by means of a simple argument observed in \cite{Gr};
see Theorem~\ref{thm-lower-bd} below.

First, we consider the Dirac operator in which case
the assertion of
the following theorem is more or less folkloric.
Its proof below extends, however, easily to
the non-local no-pair operator.
For any subset $I\subset(-1,1)$ we introduce the notation
\begin{equation}\label{def-deltaI}
\delta(I)\,:=\,\inf\{\,\sqrt{1-\lambda^2}:\,\lambda\in I\,\}\,.
\end{equation}

\begin{theorem}\label{thm-L2-exp-Dirac}
Assume that $V$ fulfills Hypothesis~\ref{hyp-Vogel}
with $\gamma\in[0,1)$
and that $A\in L^\infty_\loc(\RR^3,\RR^3)$ and
let $I\subset(-1,1)$ be some compact interval.
Then, for every $a<\delta(I)$,
there exists a constant $C(a,I)\in(0,\infty)$
such that, for all $A\in L^\infty_\loc(\RR^3,\RR^3)$,
\begin{equation}\label{L2-est-Dirac}
\big\|\,e^{a|x|}\,E_I(\D{A,V})\,\big\|\,\klg\,C(a,I)\,.
\end{equation}
\end{theorem}

\begin{proof}
First, we fix $a\in(0,\delta(I))$,
pick some cut-off function $\chi\in C^\infty(\RR^3,[0,1])$
such that $\chi(x)=0$, for $|x|\klg1$,
and $\chi(x)=1$, for $|x|\grg2$,
and set $\chi_R(x):=\chi(x/R)$, $x\in\RR^3$, $R\grg1$.
By the monotone convergence theorem
it suffices to show that
$$
\big\|\,\chi_{2R}\,e^F\,E_I(\D{A,V})\,\big\|\,\klg\,\const(a,R)\,,
$$
for some $R\grg1$ and all functions $F$
satisfying 
\begin{equation}\label{hyp-F-L2}
F\in C^\infty(\RR^3,\RR)\cap L^\infty(\RR^3,\RR)\,,
\quad F(x)=0\,,\;|x|\klg1\,,\quad
F\grg0\,,
\quad
|\nabla F|\,\klg\,a\,,
\end{equation}
To this end we introduce
$$
V_R\,:=\,\chi_R\,V\,, \qquad R\grg1\,,
$$
and pick some $\ve>0$ such that it still holds
$a<\delta(I_\ve)$, where $I_\ve:=I+(-\ve,\ve)$. 
Choosing $R\grg1$ sufficiently large we may
assume in the following that every
$z\in I_\ve+i\RR$ belongs to the resolvent set
of $\DA+i\alpha\cdot\nabla F+V_R$, for every $F$
satisfying \eqref{hyp-F-L2} (in particular $F=0$).
Using the notation \eqref{def-RAV}, we may further assume that 
\begin{eqnarray}
C(a,R)&:=&\sup\big\{\,\nonumber
\|\,e^F\,\R{A,V_R}{z}\,e^{-F}\,\|\,:
\\
& &\qquad\label{def-CaR}
\:z\in I_\ve+i\RR\,,\;A\in L_\loc^\infty(\RR^3,\RR^3)\,,\;
F\;\textrm{satisfies \eqref{hyp-F-L2}.}\,\big\}
\,<\,\infty\,.
\end{eqnarray}
In fact, since $\|V_R\|\to0$, $R\to\infty$,
this is a simple consequence of Lemma~\ref{le-marah}.
Next, we pick some $\omega\in C_0^\infty(\RR,[0,1])$
such that $\omega\equiv1$ on $I$ and
$
\supp(\omega)\subset I_\ve
$ 
and preserve the symbol $\omega$ to denote
an almost
analytic extension of $\omega$
to a smooth, compactly supported
function on the complex plane
such that
\begin{eqnarray}
&&\supp(\omega)\,\subset\,I_\ve\,+\,i(-\delta, \delta)\,\subset\,
\vr(\D{A,V_R}+i\alpha\cdot\nabla F)\,,\nonumber
\\
&& \partial_{\ol{z}}\omega(z)\,=\,\bigO_N\big(|\Im z|^N\big)\,,
\quad N\in\NN\,.\label{partialolzomega}
\end{eqnarray}
Here $\partial_{\ol{z}}=\frac{1}{2}(\partial_{\Re z}+i\partial_{\Im z})$
and
$\delta>0$ can be chosen arbitrarily.
We have $\omega(\D{A,V_R})=0$.
By virtue of the Helffer-Sj\"ostrand formula,
\begin{equation*}\label{HS-for}
\omega(T)\,=\,\int_\CC\,(T-z)^{-1}\,d\omega(z),\qquad
d\omega(z)\,:=\,
-\frac{i}{2\pi}\,\partial_{\ol{z}}\omega(z)\,dz\wedge d\ol{z},
\end{equation*}
which holds for every self-adjoint operator $T$ on some Hilbert space
(see, e.g., \cite{DiSj}; one
could also use a similar formula due to Amrein et al. 
\cite[Theorem~6.1.4(d)]{ABdMG} which avoids almost analytic extensions
but consists of a sum of integrals over resolvents),
we deduce that
\begin{eqnarray}
\chi_{2R}\,E_I(\D{A,V})&=&\big(\,\chi_{2R}\,\omega(\D{A,V})\,
-\,\omega(\D{A,V_R})\,\chi_{2R}\,\big)\,E_I(\D{A,V})\nonumber
\\
&=&
\int_\CC \big(\chi_{2R}\,\R{A,V}{z}
-\R{A,V_R}{z}\,\chi_{2R}\big)\,E_I(\D{A,V})\,
d\omega(z)\,.\nonumber
\end{eqnarray}
Since $\chi_{2R}\,(V-V_R)=0$ we infer 
by means of \eqref{res-id} that,
for all $F$ satisfying \eqref{hyp-F-L2},
\begin{eqnarray}
\lefteqn{
\chi_{2R}\,e^F\,E_I(\D{A,V})\label{sonja2}
}
\\
&=&\nonumber
-\int_\CC e^F\,\R{A,V_R}{z}\,e^{-F}\,
(e^F\,i\alpha\cdot\nabla\chi_{2R})\,
\R{A,V}{z}\,E_I(\D{A,V})\,d\omega(z)\,.
\end{eqnarray}
On account of \eqref{hyp-F-L2}, \eqref{def-CaR}, and \eqref{partialolzomega}
we thus get
$$
\big\|\,\chi_{2R}\,e^F\,E_I(\D{A,V})\,\big\|
\,\klg\,C(a,R)\,\frac{e^{4aR}\,
\|\nabla\chi\|_\infty}{2R}\int_\CC \frac{|d\omega(z)|}{|\Im z|}
\,<\,\infty\,.
$$
\end{proof}

\begin{theorem}\label{thm-L2-exp}
Assume that $V$ fulfills Hypothesis~\ref{hyp-Vogel}
with $\gamma\in[0,\gamma_{\mathrm{c}}]$
and that $A\in L^\infty_\loc(\RR^3,\RR^3)$.
If $\gamma=\gamma_{\mathrm{c}}$ 
assume further that $A$ is Lipschitz continuous
in some neighbourhood of $0$.
Let $I\subset(-1,1)$ be some compact interval and 
$a\in(0,\delta(I))$.
Then $\dom(e^{a|x|})\supset\Ran(E_I(\B{A,V}))$
there exists some $A$-independent constant $C(a,I)\in(0,\infty)$
such that, for all $\zeta\in C_0^{\infty}(\{|x|\klg\rho\},[0,1])$
with $\zeta\equiv1$ in a neighbourhood of $0$
($\rho$ is the parameter appearing in Hypothesis~\ref{hyp-Vogel}),
$$
\big\|\,e^{a|x|}\,E_I(\B{A,V})\,\big\|_{\LO(\HR^+_A,\HR)}\,\klg\,
C(a,I)\,\big(1+\|\nabla\zeta\|+\|\zeta\,A\|+\|(1-\zeta)\,V\|\big)
\,.
$$
\end{theorem}

\begin{proof}
We fix some 
$a\in(0,\delta(I))$
and define   
\begin{equation}\label{def-wtDAV}
\wt{D}_{A,V}\,:=\,\B{A,V}\oplus\DA\,\PAm\,,
\end{equation}
so that
$E_I(\wt{D}_{A,V})\,\PA=E_I(\B{A,V})\oplus 0$.  
We choose $\chi_R$, $V_R$, $\ve$, $I$, and $\omega$
in the same way as in the proof of Theorem~\ref{thm-L2-exp-Dirac}
and introduce the comparison operator
\begin{equation}\label{def-wtHH}
\wt{D}_{A,V_R}\,:=\,
\DA\,+\,\PA\,V_R\,\PA\,.
\end{equation}
Then it is clear that 
$$
\omega(\wt{D}_{A,V_R})\,\PA\,=\,0\,,
$$
for all sufficiently large $R\grg1$.
In particular, writing
\begin{equation}\label{def-wtR}
\wt{R}_{A,V}(z)\,:=\,(\wt{D}_{A,V}-z)^{-1}\,,
\qquad \wt{R}_{A,V_R}(z)\,:=\,(\wt{D}_{A,V_R}-z)^{-1}\,,
\end{equation}
we deduce the following analogue of \eqref{sonja2}
$$
\chi_{2R}\,e^F\,E_I(\wt{D}_{A,V})\,\PA\,=\,
\int_\CC e^F\,\big(\chi_{2R}\,\wt{R}_{A,V}(z)
-\wt{R}_{A,V_R}(z)\,\chi_{2R}\big)\,
\PA\,E_I(\wt{D}_{A,V})
\,d\omega(z)\,.
$$
Therefore, it suffices to show that,
for some sufficiently large $R\grg1$,
there is some $C(a,R)\in(0,\infty)$ such that, for
all $F$ satisfying \eqref{hyp-F-L2}
\begin{equation}\label{tine-1}
\sup_{z\in\supp(\omega)\setminus\RR}|\Im z|\,
\big\|\,e^F\,\big(\wt{R}_{A,V_R}(z)\,\chi_{2R}
-\chi_{2R}\,\wt{R}_{A,V}(z)\big)
\,\big\|\,\klg\,C(a,R)\,.
\end{equation}
To this end we first remark that due to 
\eqref{peki2}, \eqref{sonja99}, and $\|V_R\|\to0$, $R\to\infty$,
we find some constant $C'(a,R)\in(0,\infty)$
such that, for all $A\in L_\loc^\infty(\RR^3,\RR^3)$
and all $F$ satisfying \eqref{hyp-F-L2},
\begin{eqnarray}
\lefteqn{\nonumber
\big\|\,\big[\,\chi_{2R}\,,\,\PA\,V_R\,\PA\,\big]\,e^F\,\big\|
}
\\
&\klg&\label{tine0}
\big\|\,\big([\chi_{2R}\,,\,\PA]\,e^F\big)\,V_R\,\PAFm\,\big\|
\,+\,
\big\|\,\PA\,V_R\,\big([\chi_{2R}\,,\,\PA]\,e^F\big)\,\big\|
\,\klg\,C'(a,R)\,.
\end{eqnarray}
Writing $\ol{\chi}_{R}:=1-\chi_R$ we further observe that
$$
V\,\ol{\chi}_R\,\PA\,e^F\,\chi_{2R}
\,=\,
(\id_{\{|x|\klg2R\}}\,e^F)\,V\,\big[\,e^{-F}\,\ol{\chi}_R\,,\,\PA\,\big]
\,e^F\,\chi_{2R}\,,
$$
which together with \eqref{peki} and \eqref{Vpeki} implies 
\begin{equation}\label{tine1}
\big\|\,V\,\ol{\chi}_R\,\PA\,e^F\,\chi_{2R}\,\big\|\,\klg\,
C''(a,R)\,\big(\|\nabla\zeta\|+\|\zeta\,A\|+\|(1-\zeta)\,V\|\big)\,,
\end{equation}
for some constant $C''(a,R)\in(0,\infty)$ which neither depends
on $A$ nor $\zeta$.
Now, a straightforward computation yields, for $\vp\in\COI$
and $z\in\CC\setminus\RR$,
\begin{eqnarray}
\lefteqn{\nonumber
\big(\,\chi_{2R}\,\wt{R}_{A,V_R}(z)-\wt{R}_{A,V}(z)\,\chi_{2R}\,\big)\,
(\wt{D}_{A,V_R}-z)\,\vp
}
\\
&=&\label{sonja77}
\wt{R}_{A,V}(z)\,\PA\,V\,\ol{\chi}_R\,\PA\,\chi_{2R}\,\vp\,-\,
\wt{R}_{A,V}(z)\,i\alpha\cdot\nabla\chi_{2R}\,\vp
\\
& &\;-\,\nonumber
\wt{R}_{A,V}(z)\,
[\chi_{2R}\,,\,\PA\,V_R\,\PA]
\,\vp\,.
\end{eqnarray}
Since the range of $(\wt{D}_{A,V_R}-z)\!\!\upharpoonright_{\COI}$
is dense this together with \eqref{tine1} implies
\begin{eqnarray}
\lefteqn{\nonumber
\big(\chi_{2R}\,\wt{R}_{A,V_R}(z)-\wt{R}_{A,V}(z)\,\chi_{2R}\big)\,e^F
}
\\
&=&
\wt{R}_{A,V}(z)\,\PA\,V\,\ol{\chi}_R\,\PA\,e^F\,
\chi_{2R}\,\big(e^{-F}\,\wt{R}_{A,V_R}(z)\,e^F\big)\nonumber
\\
& &
-\,\wt{R}_{A,V}(z)\,
\big\{\,i\alpha\cdot\nabla\chi_{2R}\,e^F+
[\chi_{2R}\,,\,\PA\,V_R\,\PA]\,e^F\label{tine2}
\,\big\}
\,\big(e^{-F}\,\wt{R}_{A,V_R}(z)\,e^F\big)\,.
\end{eqnarray}
Since $\|\PA\,V_R\,\PA\|\to0$, $R\to\infty$, Lemma~\ref{le-marah}
ensures that, for sufficiently large $R\grg1$,
the norm of $e^{-F}\,\wt{R}_{A,V_R}(z)\,e^F$
is uniformly bounded, for all $z\in\supp(\omega)$, 
$A\in L_\loc^\infty(\RR^3,\RR^3)$,
and every $F$ satisfying \eqref{hyp-F-L2}.
Taking the adjoint of
\eqref{tine2} and using
\eqref{tine0} and \eqref{tine1} 
we thus obtain \eqref{tine-1}.
\end{proof}

\begin{theorem}\label{thm-L2-exp-klg0}
Assume that $V$ fulfills Hypothesis~\ref{hyp-Vogel} 
with $\gamma\in(0,\gamma_{\mathrm{c}}]$
and that $A\in L^\infty_\loc(\RR^3,\RR^3)$.
Assume further that $A$ is locally Lipschitz continuous
if $\gamma=\gamma_{\mathrm{c}}$. 
Then, for every $a\in[0,1)$, 
$\Ran(E_{(-\infty,0)}(\BAVC))\subset\dom(e^{a|\cdot|})$
and there is some $A$-independent
$C(a)\in(0,\infty)$
such that, for all $\zeta\in C_0^{\infty}(\{|x|\klg\rho\},[0,1])$
with $\zeta\equiv1$ in a neighbourhood of $0$,
$$
\big\|\,e^{a|\cdot|}\,E_{(-\infty,0)}(\BAVC)\,\big\|_{\LO(\HR^+_A,\HR)}
\,\klg\,
C(a)\,\big(1+\|\nabla\zeta\|+\|\zeta\,A\|+\|(1-\zeta)\,V\|\big)
\,.
$$
\end{theorem}

\begin{proof}
We fix $a\in[0,1)$.
It follows from Theorem~\ref{thm-L2-exp} and 
Theorem~\ref{thm-lower-bd} below that  
the spectrum of $\B{A,V}$ in $(-1,1)$ is discrete,
\begin{equation}
\label{gina}
\spec(\BAVC)\cap(-1,1)\,\subset\,\specd(\BAVC)\,.
\end{equation}
In particular,
we find some $e_0\in(-1,0)\cap\vr(\BAVC)$ 
such that $E_{(-\infty,0)}(\BAVC)=E_{(-\infty,e_0]}(\BAVC)$
and $1-a^2-e_0^2>0$.
Using the notation \eqref{def-wtDAV} and \eqref{def-wtHH}
we have
$$
E_{(-\infty,0)}(\wt{D}_{A,V})\,=\,E_{(-\infty,e_0]}(\wt{D}_{A,V})\,,
\qquad 
E_{(-\infty,e_0]}(\wt{D}_{A,V_R})\,=\,\PAm\,,
$$
provided $R\grg1$ is sufficiently large.
Thanks to \eqref{peki2} we know that, for fixed $R$, 
$e^F\,\PA\,(1-\chi_{2R})=\PAF\,e^F\,(1-\chi_{2R})$
is uniformly bounded, for all $F$ satisfying \eqref{hyp-F-L2}.
It thus remains to consider
\begin{eqnarray}
\lefteqn{\nonumber
e^F\,\PA\,\chi_{2R}\,\big(E_{(-\infty,e_0]}(\BAVC)\oplus0\big)
}
\\
&=&\nonumber
e^F\,\PA\,\big(\chi_{2R}\,E_{(-\infty,e_0]}(\wt{D}_{A,V})-
E_{(-\infty,e_0]}(\wt{D}_{A,V_R})\,\chi_{2R}\big)\,\PA
\\
&=&\label{maria100}
\frac{1}{2}\:e^F\,\PA\,\big(\sgn\big[\wt{D}_{A,V_R}-e_0\big]\,\chi_{2R}
-
\chi_{2R}\,\sgn\big[\wt{D}_{A,V}-e_0\big]\big)\,\PA
\,.
\end{eqnarray}
Using \eqref{Cauchy-PV} and \eqref{def-wtR} to represent 
the sign function of 
$\wt{D}_{A,V}-e_0$ and $\wt{D}_{A,V_R}-e_0$
by a strongly convergent Cauchy principal value
and using \eqref{maria100} 
we obtain, for all $\psi\in\HR$,
\begin{eqnarray*}
\lefteqn{
\big\|\,e^F\,\PA\,\chi_{2R}\,\big(E_{(-\infty,0)}(\BAVC)\oplus0
\big)\,\psi\,\big\|
}
\\
&\klg&
\int_{\RR}\Big\|\,\PAF\,e^F\,\big(\wt{R}_{A,V_R}(e_0+iy)\,\chi_{2R}
-\chi_{2R}\,\wt{R}_{A,V}(e_0+iy)\big)\,\PA\,\psi\,\Big\|\,\frac{dy}{2\pi}
\end{eqnarray*}
If $\delta_0>0$ denotes the distance from $e_0$ to the spectrum of
$\wt{D}_{A,V}$, then we have 
$\|\wt{R}_{A,V}(e_0+iy)\|=(\delta_0^2+y^2)^{1/2}$.
A straight-forward Neumann expansion employing \eqref{marah1} and 
$$
\PAF\,V_R\,\PAF\,\longrightarrow\,0\,,\qquad R\to\infty\,,
$$
further shows that, for every sufficiently large $R\grg1$, there is 
some $C'(a,R)\in(0,\infty)$ such that, for all
$A\in L_\loc^\infty(\RR^3,\RR^3)$ and all $F$ satisfying \eqref{hyp-F-L2},
$$
\big\|\,e^F\,\wt{R}_{A,V_R}(e_0+iy)\,e^{-F}\,\big\|\,\klg\,
\frac{C'(a,R)}{\sqrt{1+y^2}}\,,\qquad y\in\RR\,.
$$
Using \eqref{peki2}, \eqref{tine0}, \eqref{tine1}, 
and \eqref{tine2}, we thus
arrive at
$$
\big\|\,e^F\,\PA\,\chi_{2R}\,\big(E_{(-\infty,0)}(\BAVC)\oplus0
\big)\,\psi\,\big\|
\,\klg\,C''(a,R)\,
\big(1+\|\nabla\zeta\|+\|\zeta\,A\|+\|(1-\zeta)\,V\|\big),
$$
for all $\psi\in\HR$, $\|\psi\|=1$,
where the constant 
$C''(a,R)\in(0,\infty)$ neither depends on $A$ nor $\zeta$.
\end{proof}


\section{The discrete and essential spectra of $\B{A,V}$}

\noindent
Next, we consider the discrete and essential spectra
of $\B{A,V}$. To start with we
prove a theorem we have already refered to in the proof
of Theorem~\ref{thm-L2-exp-klg0} (to obtain \eqref{gina})
and which completes our proof of Theorem~\ref{main-thm-L2}.
It is used to infer the lower bound on the essential spectrum
of $\B{A,V}$ from our localization estimates and proved
by adapting an argument we learned from \cite{Gr}
to the non-local no-pair operator.
Certainly, one could also try to locate the essential spectrum
of $\B{A,V}$ by a more direct method without relying on
exponential localization estimates.
We refer to \cite{LaSi} for recent developments
relevant to this question
and numerous references.

\begin{theorem}\label{thm-lower-bd}
Assume that $V$ fulfills Hypothesis~\ref{hyp-Vogel}
with $\gamma\in[0,\gamma_{\mathrm{c}}]$ and let 
$A\in L_\loc^\infty(\RR^3,\RR^3)$.
If $\gamma=\gamma_{\mathrm{c}}$ assume further that $A$ is locally
Lipschitz continuous.
Let $I\subset(-\infty,1)$ be an interval such that
$\Ran(E_I(\B{A,V}))\subset\dom(e^{\ve|\cdot|})$, for
some $\ve>0$.
Then the spectral projection $E_I(\B{A,V})$
is a compact and, hence, finite rank operator on $\HR^+_A$.
\end{theorem}

\begin{proof}
We pick some cut-off function $\chi\in C^\infty(\RR^3,[0,1])$
such that $\chi(x)=1$, for $|x|\klg1$,
and $\chi(x)=0$, for $|x|\grg2$,
and set $\chi_R(x):=\chi(x/R)$, $x\in\RR^3$, $R\grg1$.
Since
$e^{\ve|x|}\,E_I(\B{A,V})\in\LO(\HR^+_A,\HR)$
and since $\chi_R\,e^{-\ve|x|}\to e^{-\ve|x|}$, $R\to\infty$, in the operator
norm, it suffices to show that
$\chi_R\,E_I(\B{A,V})=\chi_R\,e^{-\ve|x|}\,e^{\ve|x|}\,E_I(\B{A,V})$
is compact, for every $R\grg1$.
First, we show this assuming that $A\in L^\infty_\loc(\RR^3,\RR^3)$
and that $V$ fulfills Hypothesis~\ref{hyp-Vogel} with
$\gamma\in[0,\gamma_{\mathrm{c}})$.

Since $\DA$ has the local compactness property
we know that $\chi_R\,|\DA|^{-1/2}$ is compact, for all $R\grg1$. 
It thus remains to show that
$|\DA|^{1/2}\,E_I(\B{A,V})\in\LO(\HR^+_A)$,
which in turn is readily proved writing
\begin{equation}\label{judith}
|\DA|^{1/2}\,E_I(\B{A,V})
\,=\,\big\{|\DA|^{1/2}\,\PA\,(\B{A,V}+c)^{-1/2}\big\}
\,(\B{A,V}+c)^{1/2}\,E_I(\B{A,V})\,,
\end{equation}
where $c>-\inf\spec(\B{A,V})$.
In fact, by Theorem~\ref{thm-BRA-bb} the form domain of $\B{A,V}$
is $\PA\,\dom(|\DA|^{1/2})$ and, hence, the operator $\{\cdots\}$
in \eqref{judith} is bounded.

Next, we treat the case $\gamma=\gamma_{\mathrm{c}}$ assuming
that $A$ is locally Lipschitz continuous.
In this case Remark~\ref{rem-B0V-BAV} is applicable
and we may represent 
$\chi_R\,E_I(\B{A,V})=\chi_{2R}\,\chi_R\,\PA\,E_I(\B{A,V})$ as
\begin{eqnarray}
\lefteqn{
\chi_R\,E_I(\B{A,V})\,=\,
(\chi_{2R}\,|\D{0}|^{-\vk})\,\big\{\,|\D{0}|^{\vk}\,\PO
\B{0,-\gamma_{\mathrm{c}}/|\cdot|}^{-1/2}\big\}\,\times\label{sheila1}
}
\\
& &\;\;
\times\,\label{sheila2}
\big\{\B{0,-\gamma_{\mathrm{c}}/|\cdot|}^{1/2}\,\PO\,
\chi_R\,(\B{A,V}+c)^{-1/2}\big\}\,(\B{A,V}+c)^{1/2}\,E_I(\B{A,V})
\\
& &+\;\label{sheila3}
\chi_{2R}\,|\D{0}|^{-1/2}\,\big\{\,|\D{0}|^{1/2}\,
\,(\chi_R\,\PA-\PO\,\chi_R)\big\}\,E_I(\B{A,V})\,,
\end{eqnarray}
for some $\vk\in(0,1/4)$.
We recall from \cite{Tix2}
that $\dom(\B{0,-\gamma_{\mathrm{c}}/|\cdot|})\subset\dom(|\D{0}|^s)$, 
for every
$s\in(0,1/2)$
and, hence, 
$\dom(\B{0,-\gamma_{\mathrm{c}}/|\cdot|}^{1/2})\subset\dom(|\D{0}|^\vk)$.
Therefore, the operator $\{\cdots\}$ in \eqref{sheila1} is bounded.
The operator $\{\cdots\}$ in \eqref{sheila2} is bounded
because of Remark~\ref{rem-B0V-BAV}, and the one in curly
brackets in \eqref{sheila3} is bounded according to \eqref{rudi0}.
Since $\chi_{2R}\,|\D{0}|^{-s}$ is compact, for all $s>0$,
the theorem is proved.
\end{proof}

\bigskip

In the remaining part of this section we prove Theorem~\ref{thm-ex-no-pair}.

\begin{lemma}\label{thm-upper-bd}
Assume that $V$ fulfills Hypothesis~\ref{hyp-Vogel}
with $\gamma\in[0,\gamma_{\mathrm{c}}]$
and that $A$ fulfills Hypothesis~\ref{hyp-A-ex}(i).
Let $\lambda\in[1,\infty)$ and let $\{\psi_n(\lambda)\}_{n\in\NN}$
denote the Weyl sequence appearing in
Hypothesis~\ref{hyp-A-ex}(i).
Then
\begin{eqnarray}
\|\PA\,\psi_n(\lambda)\|&\longrightarrow&1\,,\label{jurek1}
\qquad n\to\infty\,,
\\
\|(\B{A,V}-\lambda)\,\PA\,\psi_n(\lambda)\|
&\longrightarrow&0\,,\qquad n\to\infty\,.\label{jurek2}
\end{eqnarray}
\end{lemma}

\begin{proof}
Since $(\DA-\lambda)\,\psi_n(\lambda)\to0$ and 
$\|\psi_n(\lambda)\|=1$, \eqref{jurek1} follows from the 
the spectral calculus; see \cite[Lemma~6.2]{MaSt08a}.
Next, we pick some $\vt\in C^\infty(\RR^3,[0,1])$
such that $\vt(x)=0$, for $|x|\grg1/2$,
and $\vt(x)=1$, for $|x|\grg1$,
and set $\vt_R:=\vt(\cdot/R)$, $R\grg1$. Then
$\psi_n(\lambda)=\vt_{R_n}\,\psi_n(\lambda)$
and, hence,
$$
V\,\PA\,\psi_n(\lambda)
\,=\,\vt_{R_n}\,V\,\PA\,\psi_n(\lambda)
\,+\,V\,[\PA\,,\,\vt_{R_n}]\,\psi_n(\lambda)\,.
$$
In view of Hypothesis~\ref{hyp-Vogel}
and \eqref{Vpeki}
we thus have $\|\PA\,V\,\PA\,\psi_n(\lambda)\|\to0$
and, consequently,
\eqref{jurek2} holds true also.
\end{proof}

\begin{theorem}\label{thm-eigenvalues}
Assume that $V$ fulfills Hypothesis~\ref{hyp-Vogel} and \eqref{V=O(1/R)}
and that $A$ fulfills Hypothesis~\ref{hyp-A-ex}(ii).
Then $\B{A,V}$ has infinitely many eigenvalues
below $1=\inf\specess(\B{A,V})$.
\end{theorem}

\begin{proof}
We construct appropriate trial functions
by means of the Weyl sequence
$\{\psi_n(1)\}_{n\in\NN}$
of Hypothesis~\ref{hyp-A-ex}(ii).
It is shown in \cite[Lemma~7.7]{MaSt08a}
that, for every $d\in\NN$, there is some $n_0\in\NN$
such that the set of vectors
$\{\PA\,\psi_n(1)\}_{n=m_0}^{m_0+d}$
is linearly independent, for all $m_0\in\NN$, $m_0\grg n_0$.
Setting
\[
\Psi\,:=\,\sum_{n=m_0}^{m_0+d}c_n\,\PA\,\psi_n(1)\,,
\]
for $c_{m_0},\ldots,c_{m_0+d}\in\CC$, 
we clearly have
\begin{eqnarray}
\lefteqn{
\SPn{\Psi}{(\B{A,V}-1)\,\Psi}\nonumber
}
\\
&\klg&\label{sergey1}
\,\sum_{n=m_0}^{m_0+d}|c_n|^2
\,\SPb{\PA\,\psi_n(1)}{(\DA-1+V)\,\PA\,\psi_n(1)}
\\
& &\,+\;\label{sergey2}
\sum_{{n,m=m_0\atop n\not=m}}^{m_0+d}|c_n|\,|c_m|
\,\big|\SPb{\PA\,\psi_n(1)}{(\DA-1+V)\,\PA\,\psi_m(1)}\big|\,.
\end{eqnarray}
We first comment on the terms in \eqref{sergey1}.
Employing the fact that the lower two spinor components
of $\psi_n(1)$ vanish, for all $n\in\NN$,
it is shown in \cite[Lemma~7.1]{MaSt08a}
that there is some $C\in(0,\infty)$ such that
\begin{equation}\label{august1}
0\,\klg\,\SPb{\PA\,\psi_n(1)}{(\DA-1)\,\PA\,\psi_n(1)}
\,\klg\,C\,R_n^{-2}\,,\qquad n\in\NN\,.
\end{equation}
Moreover, we find some constant $C'\in(0,\infty)$
such that, for all $n\in\NN$,
\begin{equation}\label{august2}
\SPb{\PA\,\psi_n(1)}{V\,\PA\,\psi_n(1)}
\,\klg\,-\frac{\wt{\gamma}}{(1+2\delta)\,R_n}
\,\big\|\,\PA\,\psi_n(1)\,\big\|^2
\,+\,C'\,e^{-R_n/C'}\,.
\end{equation}
In fact, 
since the quadratic form $V(x)$ is negative 
it clearly suffices to prove \eqref{august2}
with $V$ replaced by $V_\mathrm{r}:=\id_{\{|x|\grg1\}}\,V$.
Then its proof is, however, exactly the same
as the one of \cite[Lemma~7.3]{MaSt08a}.
(Just replace $\PAV$ by $\PA$ there.)
The terms in \eqref{sergey2} are treated in
Lemma~\ref{le-cross-terms} below, where we show that
\begin{equation}\label{august3}
\big|\SPb{\PA\,\psi_n(1)}{(\DA-1+V)\,\PA\,\psi_m(1)}\big|
\,=\,\bigO(R_n^{-\infty})\,,\qquad m>n\,,
\end{equation}
as $n$ tends to infinity.
Combining \eqref{jurek1} and \eqref{august1}-\eqref{august3} with
Hypothesis~\ref{hyp-Vogel} we find some $\delta_0>0$
such that
$$
\SPn{\Psi}{(\B{A,V}-1)\,\Psi}\,\klg\,-\delta_0
\sum_{n=m_0}^{m_0+d}|c_n|^2\,,
$$
for all $c_{m_0},\ldots,c_{m_0+d}\in\CC$,
provided $m_0\in\NN$ is sufficiently
large (depending on $d$). This implies the assertion
of the theorem.
\end{proof}

\begin{lemma}\label{le-cross-terms}
Assertion \eqref{august3} holds
under the assumptions of Theorem~\ref{thm-eigenvalues}.
\end{lemma}

\begin{proof}
We pick a family of smooth weight functions, 
$\{F_{k\ell}\}_{k,\ell\in\NN}$,
such that $F_{k\ell}\equiv0$ on $\supp(\psi_k(1))$,
$F_{k\ell}$ is constant on $\{|x|\klg1\}$
and outside some ball containing
$\supp(\psi_k(1))$ and $\supp(\psi_\ell(1))$, 
$\|\nabla F_{k\ell}\|_\infty\klg a<1$, and
$$
g_{k\ell}\,:=\,\|e^{-F_{k\ell}-F_{\ell k}}\|_\infty
\,\klg\,C\,e^{-a'\,\min\{R_k,R_\ell\}}\,,
\qquad k,\ell\in\NN\,,
$$
where $a,a'\in(0,1)$ and $C\in(0,\infty)$ do not depend
on $k,\ell\in\NN$.
Such a family exists
because of \eqref{supp-psin-ex}.
We then have
\begin{eqnarray*}
\lefteqn{
\big|\SPb{\PA\,\psi_n(1)}{(\DA-1)\,\PA\,\psi_m(1)}\big|
}
\\
&\klg&
\big\|\,e^{-F_{mn}-F_{nm}}\,\psi_n(1)\,\big\|
\,
\big\|\,e^{F_{mn}}\,\PA\,e^{-F_{mn}}\,\big\|\,
\big\|\,(\DA-1)\,\psi_m(1)\,\big\|\,\klg\,g_{nm}\,C\,R_m^{-1}\,,
\end{eqnarray*}
where $C\in(0,\infty)$ neither depends on $n$ nor $m$.
In order to treat the term involving $V$
we let $\{\vt_n\}_{n\in\NN}$ denote the
sequence of cut-off functions
constructed in the proof of Lemma~\ref{thm-upper-bd}.
Then $(1-\vt_n)\,\psi_n=0$, $\|\vt_n\,V\|\klg C'$ and, 
applying \eqref{Vpeki},
we find some $C''\in(0,\infty)$ such that,
for all $n,m\in\NN$,
\begin{eqnarray*}
\lefteqn{
\big|\SPb{\psi_n(1)}{\PA\,V\,\PA\,\psi_m(1)}\big|
}
\\
&\klg&
g_{nm}\,\big\|\,e^{F_{mn}}\,\PA\,e^{-F_{mn}}\,\big\|\,
\big\|\,\vt_n\,V\,\big\|\,
\big\|\,e^{F_{mn}}\,\PA\,e^{-F_{mn}}\,\psi_m(1)\,\big\|
\\
& &+\,
g_{nm}\,\big\|\,e^{F_{mn}}\,\PA\,e^{-F_{mn}}\,\big\|\,
\big\|\,V\,[(1-\vt_n)\,e^{F_{mn}},\,\PA]\,e^{-F_{mn}}\,\psi_m(1)\,\big\|
\,\klg\,C''\,g_{nm}\,.
\end{eqnarray*}
\end{proof}


\section{Pointwise exponential decay}

\noindent
To begin with we construct a family of cut-off functions
which is used throughout this section.
Let $\theta\in C^\infty(\RR,[0,1])$
satisfy $\theta\equiv0$ on $(-\infty,1]$ and
$\theta\equiv1$ on $[2,\infty)$.
For $r\in(0,1/2)$ and $R\grg1$, we define 
$\chi\equiv\chi_{r,R}\in C_0^\infty(\RR^3,[0,1])$ 
by
\begin{equation}\label{def-chi}
\forall\;x\in\RR^3\::\qquad\chi(x)\,:=\,\chi_{r,R}(x)\,:=\,
\left\{
\begin{array}{ll}
\theta(|x|/r),&|x|\klg1
\\
\theta(3-|x|/R),&|x|>1.
\end{array}
\right.
\end{equation}
Then, for all $r\in(0,1/2)$ and
every multi-index
$\beta\in\NN_0^3$, we find
some constant $C(\beta,r)\in(0,\infty)$
such that
\begin{equation}\label{props-chi}
\forall\;R\grg1\::\quad
\|\partial_x^\beta\,\chi_{r,R}\|_\infty\,\klg\,C(\beta,r)\,.
\end{equation}
Furthermore, we set, for $r\in(0,1/2)$ and $x\in\RR^3$,
\begin{eqnarray}
\wt{\chi}(x)&:=&\wt{\chi}_r(x)\,:=\,
\theta(4|x|/r)\,,
\end{eqnarray}
so that $\chi=\chi\,\wt{\chi}$. 
We also fix some exponential weight function in what follows.
Let $\kappa\in C^\infty(\RR,[0,\infty))$
satisfy $\kappa\equiv0$ on $(-\infty,1]$,
$\kappa(t)=t-2$, for $t\in[3,\infty)$,
and $0\klg\kappa'\klg1$ on $\RR$.
Then we define $f\in C^\infty(\RR^3,[0,\infty))$
by
$f(x):=\kappa(|x|)$, $x\in\RR^3$,
so that
\begin{equation}\label{props-f}
\|\nabla f\|\,\klg\,1\,,\quad\forall\;\beta\in\NN_0^3\,,\;|\beta|>1\quad
\exists\;C(\beta)\in(0,\infty)\::\quad
\|\partial^\beta_xf\|_\infty\,\klg\,C(\beta)\,.
\end{equation}
On account of \eqref{exp-est-A-infty} we further have
\begin{equation}\label{hyp-A-infty-f}
\forall\;\ve>0\,,\;\beta\in\NN^3_0\;\;
\exists\;K'(\ve,\beta)\in(0,\infty)\::\quad
\big\|\,\partial_x^\beta(e^{-\ve f}\,A)\,\big\|_\infty
\,\klg\,K'(\ve,\beta)\,.
\end{equation}
In view of Sobolev's embedding theorem we shall obtain
Theorem~\ref{main-thm} as an immediate
consequence of the following result, where
$$
\|\psi\|_k\,:=\,\|\psi\|_{H^k}\,,\qquad
\psi\in H^k\,:=\,H^k(\RR^3,\CC^4)\,.
$$

\begin{theorem}\label{thm-exp-Hk}
Assume that $A$ and $V$ fulfill Hypothesis~\ref{hyp-AV-smooth}
with $\gamma\in(0,\gamma_{\mathrm{c}}]$.
Let $\phi_\lambda$ denote a normalized eigenvector of $\B{A,V}$
corresponding to an eigenvalue $\lambda\in(-\infty,1)$ and let
$\triangle$ be the function given by \eqref{def-triangle}.
Then, for all $a\in [0,\triangle(\lambda))$,
$r\in(0,1/2)$, $R\grg1$, and
$k\in\NN_0$, we have $\chi_{r,R}\,e^{af}\phi_\lambda\in H^k(\RR^3,\CC^4)$
and we find some 
$C(a,r,k)\in(0,\infty)$ such that
\begin{equation}\label{exp-Hk-est}
\forall\;R\grg1\::\quad
\big\|\,\chi_{r,R}\,e^{af}\phi_\lambda\,\big\|_k\,\klg\,C(a,r,k)\,.
\end{equation}
\end{theorem}

\begin{proof}
Of course, we prove the assertion by induction in $k\in\NN_0$.
The case $k=0$ is follows from Theorem~\ref{main-thm-L2}.
So suppose that the assertion holds true for some $k\in\NN_0$
and let
$\beta\in\NN_0^3$ be some multi-index with length $|\beta|=k$.
We pick some $a\in[0,\triangle(\lambda))$ and
set $\ve:=(\triangle(\lambda)-a)/2$, 
$\tilde{a}:=a+\ve$. 
Then the induction hypothesis together with
\eqref{hyp-V-Cinfty}, \eqref{props-chi}, \eqref{props-f},
and a simple limiting argument
implies that $\wt{\chi}\,V\,e^{\tilde{a}f}\,\phi_\lambda\in H^k$.
Using also Lemma~\ref{VPAphi-in-L2} and
$\PA\,\phi_\lambda=\phi_\lambda$,
we may thus write,
for every $\psi\in\COI$,
\begin{eqnarray}
\lambda\,\SPb{\chi\,e^{af}\,\partial^\beta_x\psi}{\phi_\lambda}
&=&\nonumber
\SPb{\chi\,e^{af}\,\partial^\beta_x\psi}{\B{A,V}\,\phi_\lambda}
\\
&=&\nonumber
\SPb{(\D{0}+\alpha\cdot A+V\,\PA)\,\chi\,e^{af}\,
\partial^\beta_x\psi}{
\PA\,\phi_\lambda}
\\
&=&\nonumber
\SPb{\D{0}\,\partial^\beta_x\psi}{\chi\,e^{af}\,\phi_\lambda}
\\
& &\;+\;
\SPb{\label{fedi1}
\partial^\beta_x\psi}{(e^{-\ve f}\,\alpha\cdot A)\,
(\chi\,e^{\tilde{a}f}\,\phi_\lambda)}
\\
& &\;-\;\label{fedi2}
\SPb{
\partial^\beta_x\psi}{i\alpha\cdot(\nabla\chi+a\,\chi\,\nabla f)\,
(e^{af}\,\wt{\chi}\,\phi_\lambda)}
\\
& &\;+\;\label{fedi3a}
\SPb{\partial^\beta_x\psi}{
\chi\,\R{0}{0}^k\,\PO\,\D{0}^k\,
(\wt{\chi}\,V\,e^{af}\,\phi_\lambda)}
\\
& &\;+\;\label{fedi3}
\SPb{\partial^\beta_x\psi}{
\chi\,(e^{af}\,\PA\,e^{-a f}-\PO)\,e^{-\ve f}\,
(\wt{\chi}\,V\,e^{\tilde{a}f}\,\phi_\lambda)}
\\
& &\;+\;\label{fedi4}
\SPb{(1-\wt{\chi})V\,e^{-af}\,\PA\,e^{af}\,\chi\,\partial^\beta_x\psi}{
e^{af}\,\phi_\lambda}
\,.
\end{eqnarray}
By the induction hypothesis, by \eqref{hyp-A-infty-f}, and by the choice of
$f$, $\chi$, and $\wt{\chi}$, it is clear that the vectors
in the right entries of the scalar products \eqref{fedi1}-\eqref{fedi3a} 
belong to $H^k$ and that their $H^k$-norms
are bounded by constants that do not depend on $R\grg1$.
Since
$\|\wt{\chi}\,V\,e^{\tilde{a}f}\,\phi_\lambda\|_k\klg\const(a,r)$,
Lemma~\ref{le-fedi3} below implies
that the right entry in \eqref{fedi3} is bounded
in $H^k$, uniformly in $R\grg1$, too. 
In order to treat the term in \eqref{fedi4} we set, for 
$\vp\in \COI$,
\begin{eqnarray*}
U\,\vp\,:=\,U_{r,R}\,\vp&:=&(1-\wt{\chi})\,V\,e^{-af}\,
\PA\,e^{af}\,\chi\,\vp
\\
&=&
V\,\big[\,(1-\wt{\chi})\,e^{-af_R}\,,\,\PA\,\big]\,e^{af_R}\,\chi\,\vp
\,.
\end{eqnarray*}
Here we are allowed to replace $f$ by some regularized weight function, 
$f_R\in C^\infty(\RR^3,[0,\infty))\cap L^\infty(\RR^3,\RR)$,
satisfying $f_R(x)=f(x)$, for $|x|\klg2R$, and $|\nabla f_R|\klg1$,
since $1-\wt{\chi}$ and $\chi$ vanish outside $\{|x|\klg2R\}$.
In view of \eqref{Vpeki} we hence know \`{a}-priori that
$U$ extends to a bounded operator on $\HR$.
Moreover, we show in Lemma~\ref{le-fedi4} below 
that 
\begin{equation}\label{Beh-U}
\exists\; C'(a,r,k)\in(0,\infty)\;\;\forall\;R\grg1\;:\quad
\|U^*\|_{\LO(\HR,H^k)}\,\klg\,C'(a,r,k)\,.
\end{equation}
Altogether this implies that the weak derivate
$(-1)^{|\beta|}\partial_x^\beta\,\D{0}\,\chi\,e^{af}\,\phi_\lambda$
exists and belongs to $\HR$
with $\HR$-norm uniformly bounded in $R\grg1$.
\end{proof}

\bigskip

\noindent
In order to prove Lemmata~\ref{le-fedi3} and~~\ref{le-fedi4}
we shall compare $e^{af}\,\RA{iy}\,e^{-af}$
with $\R{0}{iy}$. To this end
we have to regularize the difference of these two operators
by multiplying it with an exponential damping factor
(borrowed from $\phi_\lambda$ in the previous proof), as the components of
$A(x)$ might increase very quickly when $|x|$ gets large.

For $j,N\in\NN_0$, $j\klg N+1$, $a\in[0,1)$, and $\ve\in[0,1-a)$,
we abbreviate
\begin{eqnarray*}
\sA_j\,\equiv\,\sA_j^{a,\ve,N}&:=&A\,+\,i\big(a+j\,\ve/(N+1)\big)\,
\nabla f\,,
\\
\D{\sA_j}&:=&\DA\,+\,i\big(a+j\,\ve/(N+1)\big)\,
\alpha\cdot\nabla f\,,
\\
\R{\sA_j}{iy}&:=&(\D{\sA_j}-iy)^{-1}\,,\qquad y\in\RR\,.
\end{eqnarray*}
Here 
$
iy\in\vr(\D{\sA_j})
$, $y\in\RR$, because of Lemma~\ref{le-marah}, and
$\dom(\D{\sA_j})=\dom(\DA)$, since $\nabla f$ is bounded.
For $n\in\NN_0$ and 
$T_{0},\ldots,T_n\in\LO(\HR)$, we further set
\begin{equation*}
\prod_{j=0}^nT_j\, :=\, T_{0}\,T_{1}\,\cdots\,T_{n}
\,,
\qquad
\sum_{j=1}^0T_j\,:=\,0
\,.
\end{equation*}

\begin{lemma}\label{expo-xia}
Assume that $A$ fulfills Hypothesis \ref{hyp-AV-smooth} and let
$N\in\NN_0$, $a\in[0,1)$, $\ve\in[0,1-a)$, and $y\in\RR$. Then
the following identity holds true, 
\begin{eqnarray}
\lefteqn{
\big(\R{\sA_0}{iy}-\R{0}{iy}\big)\,e^{-\ve f}\nonumber
}
\\
&=&\nonumber
\sum_{k=1}^N\,(-1)^k\,
\Big\{\prod_{j=0}^{k-1}\big(\,\R{0}{iy}\,\alpha\cdot \sA_j\,
e^{-\ve f/(N+1)}\,\big)\Big\}
\,
\R{0}{iy}\,e^{-\ve(N+1-k) f/(N+1)}
\\
& &\label{exus}
+\, (-1)^{N+1}\,\Big\{
\prod_{j=0}^{N}\big(\,\R{0}{iy}\,
\alpha\cdot \sA_j\,e^{-\ve f/(N+1)}\,\big)\Big\}
\,\R{\sA_{N+1}}{iy}
\,.
\end{eqnarray}
In particular, there is some $C(k,a,\ve)\in(0,\infty)$
such that
\begin{equation}\label{HkHk}
\forall\;y\in\RR\;:
\quad 
\big\|\,\big(\R{\sA_0}{iy}-\R{0}{iy}\big)\,e^{-\ve f}
\,\big\|_{\LO(H^N)}\,\klg\, 
\frac{C(N,a,\ve)}{1+y^2}\,.
\end{equation}
\end{lemma}

\begin{proof}
We write
$g:=\ve\,f/(N+1)$ and
$z:=iy$ for short and fix some $j\in\{0,\ldots,N\}$. 
Using the argument which lead to \eqref{mildred11}
(with $\DA$ replaced by $\D{\sA_j}$ and $F=g$), 
we check that
$e^{g}\,\R{\sA_j}{z}\,e^{-g}=\R{\sA_{j+1}}{z}$. 
Now, let
$\varphi\in\HR$. 
Since $\COI$ is a core for $\DA$ and, hence, also for
$\D{\sA_{j}}$, we 
find a sequence,
$\{\psi_n\}_{n\in\NN}\in \COI^\NN$, that converges to 
$\R{\sA_{j+1}}{z}\,\vp\in\dom(\D{\sA_{j+1}})=\dom(\D{\sA_j})$ 
with respect to the graph norm of $\D{\sA_{j}}$. 
Then $\D{\sA_{j}}\, e^{-g}\psi_n \to
\D{\sA_{j}}\, e^{-g}\R{\sA_{j+1}}{z}\,\vp$, since 
$\D{\sA_{j}}\, e^{-g}\psi_n
=e^{-g}\D{\sA_{j}}\psi_n+(i\alpha\cdot\nabla g)\,e^{-g} \psi_n$
and $\D{\sA_{j}}$ is closed. 
Therefore, 
\begin{eqnarray*}
\lefteqn{
\big(\,\R{\sA_j}{z} - \R{0}{z}\,\big)\,e^{-g}\,\varphi
}
\\ 
&=&
\big(\,\R{\sA_j}{z} -
\R{0}{z}\,\big)\,(\D{\sA_j}-z)\,\R{\sA_j}{z}\,e^{-g}\,\varphi
\\ 
&=&
\big(\,\R{\sA_j}{z} -
\R{0}{z}\,\big)\,(\D{\sA_j}-z)\,e^{-g}\,\R{\sA_{j+1}}{z}\,\varphi
\\
&=&
\lim_{n\to\infty}\big(\,\R{\sA_j}{z} -
\R{0}{z}\,\big)\,(\D{0}+\alpha\cdot\sA_j-z)\,e^{-g}\psi_n
\\
&=&
-\lim_{n\to\infty}
\R{0}{z}\,(\alpha\cdot \sA_j\,e^{-g})\,\psi_n
\;=\;-\R{0}{z}\,
\alpha\cdot \sA_j\,e^{-g}\,\R{\sA_{j+1}}{z}\,\varphi
\,.
\end{eqnarray*}
Here the last step is justified according to
\eqref{exp-est-A-infty}.
The identity \eqref{exus} now follows from an obvious
combination of
\begin{eqnarray*}
\R{\sA_j}{z}\,e^{-g}
&=&  \R{0}{z}\,e^{-g}\,\varphi-\R{0}{z}\,
\alpha\cdot \sA_j\,e^{-g}\,\R{\sA_{j+1}}{z}\,,
\end{eqnarray*}
with $j=0,1,\ldots,N$.
The estimate \eqref{HkHk} follows from \eqref{exus} and
the bounds
\begin{equation}\label{fabricio1}
\|\,\R{0}{iy}\,\|_{\LO(H^\ell)}\,\klg\,(1+y^2)^{-1/2}\,,\quad
\|\,\R{0}{iy}\,\|_{\LO(H^\ell,H^{\ell+1})}\,\klg\,1\,,
\end{equation}
where $\ell\in\NN_0$,
\begin{equation}\label{fabricio2}
\|\,\R{\sA_{N+1}}{iy}\,\|_{\LO(\HR)}\,\klg\,\frac{\sqrt{3}}{\sqrt{1+y^2}}
\,\frac{\sqrt{1+(a+\ve)^2}}{\sqrt{1-(a+\ve)^2}}\,,
\end{equation}
which is a special case of \eqref{marah1}, and
\begin{eqnarray}
\label{fabricio3}
\big\|\,e^{-\ve(N+1-k)f/(N+1)}\,\big\|_{\LO(H^\ell)}&\klg&
C'(k,\ell,N,\ve)\,,
\\
\label{fabricio4}
\big\|\,\alpha\cdot\sA_j\,e^{-\ve f/(N+1)}\,\big\|_{\LO(H^\ell)}&\klg&
C''(j,\ell,N,\ve)\,,
\end{eqnarray}
which hold true by construction of $f$ and \eqref{exp-est-A-infty}.
\end{proof}

\begin{corollary}\label{cor-fedi1}
Let 
$\chi_{-1}^<,\chi\in C^\infty(\RR^3,[0,1])$ satisfy 
$\chi_{-1}^<(x)=0$, for $|x|\grg1$,
and
$\dist(\supp(\chi),\supp(\chi_{-1}^<))>0$.
Then, for all $N\in\NN$,
there is a constant $C(N)\in(0,\infty)$
such that, for all $y\in\RR$,
$$
\big\|\,\chi\,\R{\sA_0}{iy}\,\chi_{-1}^<\,\big\|_{\LO(\HR,H^N)}\,\klg\,
\frac{C(N)}{1+y^2}\:.
$$
\end{corollary}

\begin{proof}
We set 
$\chi_{N+1}^>:=\chi$,
and pick cut-off functions 
$\chi_0^<,\ldots,\chi^<_N\in C_0^\infty(\RR^3,[0,1])$
such that $\chi_{j}^<\equiv1$ on the support of
$\chi_{j-1}^<$ and such that
$\chi_{j}^<$ and
$\chi_{j+1}^>$ have disjoint supports,
where $\chi_j^>:=1-\chi_j^<$,
$j=0,\ldots,N$.
Since $\chi_{-1}^<(x)=0$, for $|x|\grg1$, we have
$\chi_{-1}^<=\chi_{-1}^<\,e^{-\ve f}$, and, by
construction, $\chi=\chi\,\chi_{k}^>$, $k=0,\ldots,N$.
Therefore, 
\eqref{exus} with $g:=\ve f/(N+1)$ yields
\begin{eqnarray}
\lefteqn{\nonumber
\chi\,\R{\sA_0}{iy}\,\chi_{-1}^<\,-\,\chi\,\chi_0^>\,\R{0}{iy}\,\chi_{-1}^<
}
\\
&=&
\label{clelia1}
\sum_{k=1}^N\,(-1)^k\!\!\!\sum_{\sharp_0,\ldots,\sharp_{k-1}\in\{<,>\}}
\chi\,\chi_k^>\,
\Big\{\prod_{j=0}^{k-1}\big(\,\R{0}{iy}\,\chi_j^{\sharp_j}\,
\alpha\cdot \sA_j\,
e^{-g}\,\big)\Big\}
\,
\R{0}{iy}\,\chi_{-1}^<
\\
& &\nonumber
+\, (-1)^{N+1}\,\chi\,\Big\{
\prod_{j=0}^{N}\big(\,\R{0}{iy}\,
\alpha\cdot \sA_j\,e^{-g}\,\big)\Big\}
\,\R{\sA_{N+1}}{iy}\,\chi_{-1}^<
\,.
\end{eqnarray}
Each summand in \eqref{clelia1} contains at least one factor
of the form $\chi^>_{j+1}\,\R{0}{iy}\,\chi^<_j$, and 
since $\chi^<_j$ and $\chi^>_{j+1}$
have disjoint supports
we readily
verify that
$$
\big\|\,\chi^>_{j+1}\,\R{0}{iy}\,\chi^<_j\,\big\|_{\LO(\HR,H^M)}
\,\klg\,
\frac{C(j,M)}{1+y^2}\,,\qquad M\in\NN\,.
$$
Together with the bounds \eqref{fabricio1}-\eqref{fabricio4} this
implies the asserted estimate.
\end{proof}

\begin{lemma}\label{le-fedi3}
Let $k\in\NN_0$, $r\in(0,1/2)$, $a\in[0,\triangle(\lambda))$,
and $\ve=(\triangle(\lambda)-a)/2$. Then there
is some $C(a,r,k)\in(0,\infty)$ such that
$$
\forall\;R\grg1\;:\qquad
\big\|\,\chi_{r,R}\,(e^{af}\,\PA\,e^{-a f}-\PO)
\,e^{-\ve f}\,\big\|_{\LO(H^k)}\,\klg\,C(a,r,k)\,.
$$
\end{lemma}

\begin{proof}
Since $e^{af}\,\RA{iy}\,e^{-af}=\R{\sA_0}{iy}$, $y\in\RR$,
Lemma~\ref{expo-xia} yields, for all $\vp,\psi\in H^k$,
\begin{eqnarray*}
\lefteqn{
\big|\SPb{\D{0}^k\,\vp}{
\chi_{r,R}\,(e^{af}\,\PA\,e^{-a f}-\PO)
\,e^{-\ve f}
\,\psi}\big|
}
\\
&\klg&
\int_\RR
\Big|\SPB{\D{0}^k\,\vp}{
\chi_{r,R}\,(\R{\sA_0}{iy}-\R{0}{iy})
\,e^{-\ve f}
\,\psi}\Big|\:\frac{dy}{2\pi}
\\
&\klg&
\int_\RR\|\vp\|\,\big\|\,\D{0}^k\,\chi_{r,R}\,\big\|_{\LO(H^k,\HR)}\,
\big\|\,(\R{\sA_0}{iy}-\R{0}{iy})
\,e^{-\ve f}\,\big\|_{\LO(H^k)}\,\|\psi\|_{k}\:\frac{dy}{2\pi}
\\
&\klg&
C'(a,\ve(a),k,r)\,\|\vp\|\,\|\psi\|_k
\,,
\end{eqnarray*}
where the constant is uniform in $R\grg1$.
\end{proof}

\begin{lemma}\label{le-fedi4}
Assertion~\eqref{Beh-U} holds true.
\end{lemma}

\begin{proof}
Since $\vt:=1-\wt{\chi}$ and $\chi$ have disjoint supports, we find some
$\chi_{-1}^<\in C_0^\infty(\RR^3,[0,1])$ such that
$\vt=\vt\,\chi_{-1}^<$ and $\dist(\supp(\chi),\supp(\chi_{-1}^<))>0$.
Then
\begin{eqnarray*}
\chi\,\R{\sA_0}{iy}\,\vt&=&
\chi\,\vt\,\R{0}{iy}\,+\,\chi\,
\R{\sA_0}{iy}\,\alpha\cdot(i\nabla\vt-\vt\,\sA_0)\,\R{0}{iy}
\\
&=&
\chi\,\R{\sA_0}{iy}\,\chi_{-1}^<\,
\alpha\cdot(i\nabla\vt-\vt\,\sA_0)\,\R{0}{iy}\,,
\end{eqnarray*}
which implies, for 
$\vp\in\COI$, $\psi\in\HR$, and $\psi_1,\psi_2,\ldots\in\COI$
such that $\psi_n\to\psi$, 
\begin{eqnarray*}
\lefteqn{
\big|\SPb{\D{0}^k\,\vp}{U^*\,\psi}\big|\,=\,
\big|\SPb{U\,\D{0}^k\,\vp}{\psi}\big|
}
\\
&\klg&
\limsup_{n\to\infty}
\int_\RR
\Big|\SPB{\D{0}^k\,\vp}{\chi\,\R{\sA_0}{iy}\,\chi^<_{-1}\,
\alpha\cdot(i\nabla\vt-\vt\,\sA_0)\,(\R{0}{iy}\,V)\,\psi_n\,}\Big|
\,\frac{dy}{2\pi}
\\
&\klg&
C_V\,\int_\RR\big\|\,\chi\,\R{\sA_0}{iy}\,\chi^<_{-1}\,
\big\|_{\LO(\HR,H^k)}\,dy\;\big(\|\nabla\vt\|+\|\vt\,A\|+\|\nabla f\|\big)\,
\,\|\vp\|\,\|\psi\|
\,.
\end{eqnarray*}
In the last line we used that
$\R{0}{iy}\,V$ extends to a bounded operator on $\HR$
with a norm bounded uniformly in $y\in\RR$ by some $C_V\in(0,\infty)$.
On account of Corollary~\ref{cor-fedi1} 
this proves Assertion~\eqref{Beh-U}.
\end{proof}


\appendix

\section{}
\label{app-K2}

\noindent
According to \cite[Theorem~7.2]{LL-Buch} a function 
$\psi$ belongs to $H^{1/2}(\RR^3)$ if and only if
$\psi\in L^2(\RR^3)$ and 
\begin{equation}\label{for-|D0|}
I(\psi)\,:=\,\frac{1}{4\pi^2}\int_{\RR^3}\int_{\RR^3}
\frac{|\psi(x)-\psi(y)|^2}{|x-y|^2}\,K_2(|x-y|)\,dx\,dy
\end{equation}
is finite, where $K_2$ denotes a modified Bessel function.
In this case
$I(\psi)=\|(1-\Delta)^{1/4}\,\psi\|^2-\|\psi\|^2$.
Now, let $\chi:\RR^3\to\CC$ be such that
$\|\chi\|_\infty+L_\chi<\infty$, where
$$
L_\chi\,:=\,\sup_{x\not=y}\frac{|\chi(x)-\chi(y)|}{|x-y|}\,.
$$
Estimating
$$
\frac{\big|(\chi\,\psi)(x)-(\chi\,\psi)(y)\big|^2}{|x-y|^2}
\,\klg\,2|\psi(x)|^2\,\frac{|\chi(x)-\chi(y)|^2}{|x-y|^2}\,+\,
2\,|\chi(y)|^2\,\frac{|\psi(x)-\psi(y)|^2}{|x-y|^2}
$$
we obtain, for every $\psi\in H^{1/2}(\RR^3)$,
$$
I(\chi\,\psi)\,\klg\,2\,L_\chi^2\,\|\psi\|^2\,\int_{\RR^3}K_2(|y|)\,dy
\,+\,2\,\|\chi\|^2_\infty\,I(\psi)\,<\,\infty\,,
$$
that is, $\chi\,\psi\in H^{1/2}(\RR^3)$,
and, using $\int_0^\infty r^2\,K_2(r)\,dr=3\pi/2$, we further obtain
\begin{eqnarray*}
\big\|\,(1-\Delta)^{1/4}\,\chi\,\psi\,\big\|^2
&=&I(\chi\,\psi)+\|\chi\,\psi\|^2
\\
&\klg&
3\,L^2_\chi\,\|\psi\|^2+\|\chi\|_\infty^2\,
\big(2\,\big\|\,(1-\Delta)^{1/4}\,\psi\,\big\|^2-\|\psi\|^2\big)\,.
\end{eqnarray*}


\bigskip

\noindent
{\bf Acknowledgement:} 
It is a pleasure to thank 
Hubert Kalf, Sergey Morozov, and Heinz Siedentop 
for helpful discussions and useful remarks.
This work has been partially supported
by the DFG (SFB/TR12).


\end{document}